\let\accentvec\vec 

\documentclass[runningheads,a4paper]{llncs}

\usepackage{fancyvrb}

\usepackage{boolean}
\usepackage{reversion}

\usepackage{amssymb}
\let\vec\accentvec 
\usepackage{amsmath}

\def\doShort
  {\let\Long\False \let\Short\True \let\Medium\False \let\Fixpoint\False}
\def\doMedium
  {\let\Long\False \let \Short\False \let\Medium\True \let\Fixpoint\True}
\def\doLong
  {\let \Long\True  \let\Short\False \let\Medium\False \let\Fixpoint\True}

\usepackage[usenames,dvipsnames]{xcolor}

\usepackage{listings}
\newcommand{\keywords}[1]{\par\addvspace\baselineskip
\noindent\keywordname\enspace\ignorespaces#1}

\usepackage{mathpartir}
\usepackage{nohyperref}  
\usepackage{url}  

\renewcommand {\RefTirName}[1]{\hypertarget{#1}{\TirName {#1}}}

\let \Rule \DefTirName
\newenvironment{mathparfig}[3][t]
   {\begin{figure}[#1]\abovedisplayskip 0em\belowdisplayskip 0em
    \def \CAPTION{\label{#2}#3}
    \def \MathparLineskip {\lineskip=0.33cm} \begin{mathpar}}
   {\end{mathpar}\caption{\CAPTION}\end{figure}}
\newcommand {\uad}{\hskip 1ex}
\newcommand{\noabovedisplayskip}{\abovedisplayskip 0em}
\newcommand {\nobelowdisplayskip}{\belowdisplayskip 0em}

\renewcommand {\RefTirName}[1]{\hypertarget{#1}{\TirName {#1}}}

\let \Rule \DefTirName

\newcommand{\code}[1]{\texttt{#1}}
\newcommand{\tyc}[1]{\mathtt{#1}}

\newcommand{\lam}[1]{\lambda #1 .}

\newcommand{\tfun}[3]{{[#1](#2)} \to #3}

\newcommand{\fix}[2]{\ensuremath{\texttt{fix} \, #1 \, #2 .}}

\newcommand{\type}[1]{\ensuremath{\mathit{#1}}}

\newcommand{\app}[0]{~}  

\newcommand{\der}{\vdash}
\newcommand{\cder}{\stackrel{c}{\der}}

\newcommand{\iI}{{i\in I}}
\newcommand{\jJ}{{j \in J}}

\newcommand{\of}{\mathpunct{:}\!}

\newcommand{\clet}[3]{\mathop{\code{let}} #1 = #2 \mathbin{\code{in}} #3}

\newcommand{\Rewto}[1]{\stackrel{#1}{\uad\leadsto\uad}}
\newcommand{\rewto}[2]{\Rewto{#1\backslash{}#2}}
\newcommand{\redto}{\implies}
\newcommand{\credto}{\stackrel{c}{\implies}}

\newcommand{\push}[1]{, #1}

\newenvironment{caselist}{%
  \begin{list}{{\it Case}}{}%
}{\end{list}%
}

\newcommand{\nextcase}{\item~}

\newcommand{\Gammap}{{\Gamma^\prime}}
\newcommand{\Phip}{{\Phi^\prime}}

\newcommand{\Phicode}[1]{{{\Phi}_{\code{#1}}}}
\newcommand{\Phidef}{\Phicode{def}}
\newcommand{\Phibody}{\Phicode{body}}
\newcommand{\Phifun}{\Phicode{fun}}
\newcommand{\Phiarg}{\Phicode{arg}}
\newcommand{\Phiclos}{\Phicode{clos}}

\newcommand{\Phipcode}[1]{{{\Phi'}_{\code{#1}}}}
\newcommand{\Phipclos}{\Phipcode{clos}}

\newcommand{\fvenv}[1]{[#1]}
\DeclareMathOperator\dom{dom}

\doShort

\begin{document}

\mainmatter  

\title{Tracking Data-Flow with Open Closure Types}

\author{Gabriel Scherer\inst{1}
        \and Jan Hoffmann\inst{2} 
       }

\institute{INRIA Paris-Rocquencourt
           \and Yale University}

\maketitle

\begin{abstract}
  Type systems hide data that is captured by function closures in
  function types.  In most cases this is a beneficial design that
  enables simplicity and compositionality.  However, some applications
  require explicit information about the data that is captured in
  closures.

  This paper introduces open closure types, that is, function
  types that are decorated with type contexts.  They are used to track
  data-flow from the environment into the function closure.  
  A simply-typed lambda calculus is used to study the properties of
  the type theory of open closure types.  A distinctive feature
  of this type theory is that an open closure type of a function
  can vary in different type contexts.
  To present an application of the type theory, it is shown that a
  type derivation establishes a simple non-interference property in
  the sense of information-flow theory.
  A publicly available prototype implementation of the system can be
  used to experiment with type derivations for example programs.
\keywords{Type Systems, Closure Types, Information Flow}
\end{abstract}

\begin{version}{\Long}
  \setcounter{tocdepth}{2}
  \makeatletter
  \makeatletter

  \tableofcontents
\end{version}

\section{Introduction}
\label{sec:intro}
Function types in traditional type systems only provide information
about the arguments and return values of the functions but not about
the data that is captured in function closures. Such function
types naturally lead to simple and compositional type systems.

Recently, syntax-directed type systems have been increasingly used to
statically verify strong program properties such as resource
usage~\cite{LagoP13,Jost10,HoffmannAH12}, information
flow~\cite{HeintzeR98,SabelfeldM03}, and
termination~\cite{Abel07,Chin01,BartheGR08}.
In such type systems, it is sometimes necessary and natural to include
information in the function types about the data that is captured by
closures.  To see why, assume that we want to design a type system to
verify resource usage.  Now consider for example the curried append
function for integer lists which has the following type in OCaml.
\[  \abovedisplayskip 0.6em
\code{append} : \type{int\ list} \to \type{int\ list} \to \type{int\ list}
\belowdisplayskip 0.6em \]
At first glance, we might say that the time complexity of
\code{append} is $O(n)$ if $n$ is the length of the first argument.
But a closer inspection of the definition of \code{append} reveals
that this is a gross simplification.  In fact, the complexity of the
partial function call \code{app\_par = append $\ell$} is constant.
Moreover, the complexity of the function \code{app\_par} is
linear---not in the length of the argument but in the length of the
list $\ell$ that is captured in the function closure.  

In general, we have to describe the resource consumption of a curried
function $f : A_1 \to \cdots \to A_n \to A$ with $n$ expressions
$c_i(a_1,\ldots,a_i)$ such that $c_i$ describes the complexity of the
computation that takes place after $f$ is applied to $i$ arguments
$a_1,\ldots,a_i$.  We are not aware of any existing type system that can
verify a statement of this form.

To express the aforementioned statement in a type system, we have to
decorate the function types with additional information about the data
that is captured in a function closure.  It is however not sufficient to
directly describe the complexity of a closure in terms of its arguments
and the data captured in the closure.  Admittedly, this would work to
accurately describe the resource usage in our example function
\code{append} because the first argument is directly captured in the
closure.  But in general, the data captured in a closure $f a_1 \cdots
a_i$ can be any data that is computed from the arguments
$a_1,\ldots,a_i$ (and from the data in the environment).  To reference
this data in the types would not only be meaningless for a user, it
would also hamper the compositionality of the type system.  It is for
instance unclear how to define subtyping for closures that capture
different data (which is, e.g., needed in the two branches of a
conditional.)

To preserve the compositionality of traditional type systems, we
propose to describe the resource usage of a closure as a function of
its argument and the data that is visible in the current environment.
To this end we introduce \emph{open closure types}, function types
that refer to their arguments and to the data in the current
environment.

More formally, consider a typing judgment of the form $\Gamma \der e :
\sigma$, in a type system that tracks fine-grained intensional
properties characterizing not only the shape of values, but the
behavior of the reduction of $e$ into a value (e.g., resource
usage). A typing rule for open closure types, $\Gamma,\Delta \der e :
\tfun{\Gamma'}{x \of \sigma}{\tau}$, captures the idea that, under a
weak reduction semantics, the computation of the closure itself, and
later the computation of the closure \emph{application}, will have
very different behaviors, captured by two different typing
environments $\Gamma$ and $\Gamma'$ of the same domain, the free
variables of $e$.  To describe the complexity of \code{append}, we
might for instance have a statement 
\[ \abovedisplayskip 0.5em
\ell \of \type{int\ list} \der \code{append}\ \ell :
\tfun{\ell \of \type{int\ list}}{y \of \type{int\ list}}{int\ list} \; .
\belowdisplayskip 0.5em
\]
This puts us in a position to use type annotations to describe the
resource usage of $\code{append}\ \ell$ as a function of $\ell$ and
the future argument $y$.  For example, using type-based amortized
analysis~\cite{HoffmannAH12}, we can express a bound on the number of
created list notes in \code{append} with the following open closure type.
\[ \abovedisplayskip 0.5em
\code{append} :
\tfun{}{ x \of \type{int\ list}^0}
{\tfun{x \of \type{int\ list}^1}{y \of \type{int\ list}^0}{int\ list}^0} \; .
\belowdisplayskip 0.5em \]
The intuitive meaning of this type for \code{append} is as follows.
To pay for the cons operations in the evaluation of $\code{append}\,
\ell_1$ we need $0 {\cdot} |\ell_1|$ resource units and to pay for the
cons operations in the evaluation of $\code{append}\, \ell_1 \,
\ell_2$ we need $0 {\cdot} |\ell_1| + 1 {\cdot}|\ell_2|$ resource
units.

The development of a type system for open closure types entails some
interesting technical challenges: term variables now appear in types,
which requires mechanisms for scope management not unlike dependent type
theories. If $x$ appears in $\sigma$, the context
$\Gamma,x\of\tau,y\of\sigma$ is not exchangeable with
$\Gamma,y\of\sigma,x\of\tau$. Similarly, the judgment $\Gamma, x \of
\tau \der e_2 : \sigma$ will not entail $\Gamma \der
\clet{x}{e_1}{e_2} : \sigma$, as the return type $\sigma$ may contain
open closures scoping over $x$, so we need to substitute variables in
types.

The main contribution of this paper is a type theory of open closure
types and the proof of its main properties.  We start from the simply-typed
lambda calculus, and consider the simple intensional property of
data-flow tracking, annotating each simply-typed lambda-calculus type
with a single boolean variable. This allows us to study the metatheory of
open closure types in clean and straightforward way.  This is the first
important step for using such types in more sophisticated type systems for 
resource usage and termination.

Our type system for data-flow tracking captures higher-order data-flow
information. As a byproduct, we get our secondary contribution, a
non-interference property in the sense of information flow theory:
high-level inputs do not influence the (low-level) results of
computations.

To experiment with of our type system, we implemented a software
prototype in OCaml (see Section~\ref{sec:impl}).
\begin{version}{\Or\Short\Medium}
  A full version of this article, containing the
  full proofs and additional details and discussion, is available
  online.\footnote{\url{http://hal.inria.fr/INRIA-RRRT/hal-00851658}}
\end{version}

\subsubsection{Related Work.} 

In our type system we maintain the invariant that open closure types
only refer to variables that are present in the current typing
context.  This is a feature that distinguishes open closure types from
existing formalisms for closure types.  

For example, while our function type
$\tfun{\Gamma^\Phi}{x \of \sigma_1}{\sigma_2}$ superficially resembles
a contextual arrow type $[\Psi](\sigma_1 \to \sigma_2)$ of contextual
type theory\cite{NanevskiPP08,PientkaD08,StampoulisS12}, we are not
aware of any actual connection in application or metatheory with these
systems. In particular, the variable in our captured context
$\Gamma^\Phi$ are \emph{bound occurrences} of the ambient typing
context, while the context $\Psi$ of a contextual type $[\Psi]T$
\emph{binds} metavariables to be used to construct inhabitants. As
such a binding can make sense in any context, our substitution
judgment has no counterpart in contextual type theory, or other modal
type theories for multi-stage programming
(\cite{Moggi99,Tsukada10}).

Having closure types carry a set of captured variables has been done in
the literature, as for example in Leroy~\cite{Leroy92}, which use
closure types to keep track of of \emph{dangerous type variables} that
can not be generalized without breaking type safety, or in the
higher-order lifetime analysis of Hannan et al.~\cite{Hannan97},
where variable sets denote variables that must be kept in memory.
However, these works have no need to vary function types in different
typing contexts and subtyping can be defined using set inclusion,
which makes the metatheory significantly simpler. On the contrary, our
scoping mechanism allows to study more complex properties, such as
value dependencies and non-interference.

The classical way to understand value capture in closures in a typed
way is through the \emph{typed closure conversion} of Minamide et
al.~\cite{MinamideMH96}.  They use existential types to account for
hidden data in function closures without losing compositionality, by
abstracting over the difference between functions capturing from
different environments. Our system retains this compositionality,
albeit in a less apparent way: we get finer-grained information about
the dependency of a closure on the ambient typing environment. 
Typed closure conversion is still possible, and could be typed in
a more precise way, abstracting only over values that are outside the
lexical context.

Petricek et al.~\cite{Petricek13} study \emph{coeffects} systems with
judgments of the form $C^r \Gamma \der e : \tau$ and function types
$C^s \sigma \to \tau$, where $r$ and $s$ are coeffect annotations over
an indexed comonad $C$. Their work is orthogonal to the present
one. They study comonadic semantics and algebraic structure of effect
indices.  These indices are simply booleans in our work but we focus on the
syntactic scoping rules that arise from tracking each variable of the
context separately.  

The non-interference property that we prove is different from the
usual treatment in type systems for information flow like the SLam
Calculus~\cite{HeintzeR98}. In SLam, the information flow into
closure is accounted for at abstraction time. In contrast, we account
for the information flow into the closure at application time.

\section{A Type System for Open Closures}
\label{sec:types}
We define a type system for the simplest problem domain that exhibits
a need for open closure types. Our goal is to determine statically,
for an open term $e$, on which variables of the environment
the value of $e$ depends.

We are interested in weak reduction, and assume a call-by-value
reduction strategy. In this context, an abstraction $\lam x e$ is
already a value, so reducing it does not depend on the environment at
all. More generally, for a term $e$ evaluating to a function
(closure), we make a distinction between the part of the environment
the reduction of $e$ depends on, and the part that will be used when
the resulting closure will be applied. For example, the term
$(y, \lam {x} z)$ depends on the variable $y$ at evaluation time, but
will not need the variable $z$ until the closure in the right pair
component is applied.

This is where we need open closure types.  Our function types are of
the form $\tfun{\Gamma^\Phi}{x \of \sigma^\phi}{\tau}$, where the
mapping $\Phi$ from variables to Booleans indicates on which variables
the evaluation depends at application time.  The Boolean $\phi$ 
indicates whether the argument $x$ is used in the function body. We 
call $\Phi$ the dependency annotation of
$\Gamma$. Our previous example would for instance be typed as follows.
\[ \abovedisplayskip 0.5em
   y \of \sigma^1, z \of \tau^0 \der (y, \lam {x} z) :
   \sigma * (\tfun{y \of \sigma^0, z \of \tau^1}{x \of \rho^0}{\tau}) 
\belowdisplayskip 0.5em\]
The typing expresses that the result of the computation depends on the
variable $y$ but not on the variable $z$.  Moreover, result of the function
in the second component of the pair depends on $z$ but not on $y$.

In general, types are defined by the following grammar.
\[ \abovedisplayskip 0.5em
\begin{array}{l@{~}l@{~}l@{\quad}r}
  \mathtt{Types} \ni \sigma,\tau,\rho & ::= & & \mbox{types} \\
      & \mid & \alpha              & \mbox{atoms} \\
      & \mid & \tau_1 * \tau_2 & \mbox{products} \\
      & \mid & \tfun{\Gamma^\Phi}{x \of \sigma^\phi}{\tau} & \mbox{closures} \\
\end{array}
\]

The closure type $\tfun{\Gamma^\Phi}{x \of \sigma^\phi}{\tau}$ binds
the new argument variable $x$, but not the variables occurring in
$\Gamma$ which are reference variables bound in the
current typing context. Such a type is \emph{well-scoped} only when
all the variables it closes over are actually present in the current
context. In particular, it has no meaning in an empty context, unless
$\Gamma$ is itself empty.

We define well-scoping judgments on contexts ($\Gamma \der$) and types
($\Gamma \der \sigma$). The judgments are defined simultaneously in
Figure~\ref{fig/scoping} and refer to each another. They use
non-annotated contexts: the dependency annotations characterize
data-flow information of \emph{terms}, and are not needed to state the
well-formedness of static types and contexts.

\begin{mathparfig}{fig/scoping}{Well-scoping of types and contexts}
\inferrule[Scope-Context-Nil]
{}
{\emptyset \der}
\and
\inferrule[Scope-Context]
{\Gamma \der \sigma}
{\Gamma, x \of \sigma \der}
\and
\inferrule[Scope-Atom]
  {\Gamma \der}
  {\Gamma \der \alpha}
\\
\inferrule[Scope-Product]
  {\Gamma \der \tau_1\\ \Gamma \der \tau_2}
  {\Gamma \der \tau_1 * \tau_2}
\and
\inferrule[Scope-Closure]
  {\Gamma_0, \Gamma_1 \der\\
   \Gamma_0 \der \sigma\\
   \Gamma_0, x \of \sigma \der \tau}
  {\Gamma_0, \Gamma_1 \der \tfun{\Gamma_0^\Phi}{x \of \sigma^\phi}{\tau}}
\end{mathparfig}

Notice that the closure contexts appearing in the return type of
a closure, $\tau$ in our rule \Rule{Scope-Closure}, may capture the
variable $x$ corresponding to the function argument, which is why we
chose the dependent-arrow--like notation $(x \of \sigma) \to \tau$
rather than only $\sigma \to \tau$. There is no dependency of types on
terms in this system, this is only used for scope tracking.

Note that $\Gamma \der \sigma$ implies $\Gamma \der$ (as proved by
direct induction until an atom or a function closure is reached). Note also that a context
type $\tfun{\Gamma_0}{x\of\sigma}{\tau}$ is well-scoped in any larger
environment $\Gamma_0, \Gamma_1$: the context information may only mention
variables existing in the typing context, but it need not mention all
of them. As a result, well-scoping is preserved by context extension: if
$\Gamma_0 \der \sigma$ and $\Gamma_0, \Gamma_1 \der$, then
$\Gamma_0, \Gamma_1 \der \sigma$.

\subsubsection{A Term Language, and a Naive Attempt at a Type System.}

Our term language, is the lambda calculus with pairs, let bindings and
fixpoints.  This language is sufficient to discuss the most
interesting problems that arise in an application of closure types in
a more realistic language.
\[
\begin{array}{l@{~}l@{~}l@{\quad}r}
  \mathtt{Terms} \ni t,u,e & ::= & & \mbox{terms} \\
      & \mid & x              & \mbox{variables} \\
      & \mid & (e_1, e_2)     & \mbox{pairs} \\
      & \mid & \pi_i(e)       & \mbox{projections ($i \in \{1,2\}$)} \\
      & \mid & \lam x e       & \mbox{lambda abstractions} \\
      & \mid & t \app u       & \mbox{applications} \\
      & \mid & \clet{x}{e_1}{e_2} & \mbox{let declarations} \\
\end{array}
\]

For didactic purposes, we start with an intuitive type system
presented in Figure~\ref{fig/naive-typing}. The judgment
$\Gamma^\Phi \der e : \sigma$ means that the expression $e$ has type
$\sigma$, in the context $\Gamma$ carrying the intensional information
$\Phi$. Context variable mapped to $0$ in $\Phi$ are not used during
the reduction of $e$ to a value. We will show that the rules
\Rule{App-Tmp} and \Rule{Let-Tmp} are not correct, and introduce a new
judgment to develop correct versions of the rules.

\begin{mathparfig}{fig/naive-typing}{Naive rules for the type system}
\inferrule[Var]
{\Gamma, x \of \sigma, \Delta \der}
{\Gamma^0, x \of \sigma^1, \Delta^0 \der x:\sigma}
\and
\inferrule[Product]
{\Gamma^{\Phi_1} \der e_1 : \tau_1\\
 \Gamma^{\Phi_2} \der e_2 : \tau_2}
{\Gamma^{\Phi_1 + \Phi_2} \der (e_1, e_2) : \tau_1 * \tau_2}
\and
\inferrule[Proj]
{\Gamma^\Phi \der e : \tau_1 * \tau_2}
{\Gamma^\Phi \der \pi_i(e) : \tau_i}
\and
\inferrule[Lam]
  {\Gamma^\Phi, x \of \sigma^\phi \der t : \tau}
  {\Gamma^0 \der \lam x t : \tfun{\Gamma^\Phi}{x \of \sigma^\phi}{\tau}}
\and
\Fixpoint{
  \inferrule[Fix]
  {\Gamma^\Phi,
    f \of (\tfun{\Gamma^\Psi}{x \of \sigma^\phi}{\tau})^\chi,
    x : \sigma^\phi
    \der e : \tau}
  {\Gamma^0 \der \fix f x e : \tfun{\Gamma^\Psi}{x \of \sigma^\phi}{\tau}}
}{}
\and
\inferrule[Let-Tmp]
  {\Gamma^\Phidef \der e_1 : \sigma\\
   \Gamma^\Phibody, x \of \sigma^\phi \der e_2 : \tau}
  {\Gamma^{\phi.\Phidef+\Phibody} \der
     \clet{x}{e_1}{e_2} : \tau}
\and
\inferrule[App-Tmp]
{(\Gamma_0, \Gamma_1)^\Phifun \der
    t :\tfun{\Gamma_0^\Phiclos
    }{x \of \sigma^\phi}{\tau}\\
 (\Gamma_0, \Gamma_1)^\Phiarg \der u : \sigma}
{(\Gamma_0, \Gamma_1)^{\Phifun+\Phiclos+\phi.\Phiarg} \der
    t \app u : \tau}
\end{mathparfig}

In a judgment
$\Gamma^0 \der \lam x t : \tfun{\Gamma^\Phi}{x \of \sigma^0}{\tau}$,
$\Gamma$ is bound only in one place (the context), and
$\alpha$-renaming any of its variable necessitates a mirroring change
in its right-hand-side occurrences ($\Gamma^\Phi$ but also in $\sigma$
and $\tau$), while $x$ is independently bound in the term and in the
type, so the aforementioned type is equivalent to
$\tfun{\Gamma^\Phi}{y \of \sigma}{\tau[y/x]}$. In particular,
variables occurring in types do \emph{not} reveal implementation
details of the underlying term.

The syntax $\phi.\Phi$ used in the \Rule{App-Tmp} and \Rule{Let-Tmp}
rules is a product, or conjunction, of the single boolean dependency
annotation $\phi$, and of the vector dependency annotation $\Phi$. The
sum $\Phi_1 + \Phi_2$ is the disjunction. In the \Rule{Let-Tmp} rule
for example, if the typing of $e_2$ determines that the evaluation of
$e_2$ does not depend on the definition $x = e_1$ ($\phi$ is $0$),
then $\phi.\Phidef$ will mark all the variables used by $e_1$ as not
needed as well (all $0$), and only the variables needed by $e_2$ will
be marked in the result annotation $\phi.\Phidef + \Phibody$.


\begin{version}{\Short}
  In the scoping judgment
  $\Gamma \der \tfun{\Gamma^\Phi}{x \of \sigma}{\tau}$, the repetition
  of the judgment $\Gamma$ is redundant. We could simply write
  $\tfun{\Phi}{x \of \sigma}{\tau}$; -- because in our simplified
  setting the intensional information $\Phi$ can be easily separated
  from the rest of the typing information, corresponding to
  simply-typed types. However, we found out that such a reformulation
  made technical developments harder to follow; the $\Gamma^\Phi$
  form allows one to keep track precisely of the domain of the
  dependency annotation, and domain changes are precisely the
  difficult technical aspect of open closure types. For a more
  detailed discussion of this design point, see the full version of
  this article.
\end{version}

\begin{version}{\Long}
In the introduction we present closure types of the form
$\tfun{\Gamma}{x \of \sigma}{\tau}$, while we here use the apparently
different form $\tfun{\Gamma^\Phi}{x \of \sigma^\phi}{\tau}$. This new
syntax is actually a simpler special case of the previous one: we
could consider a type grammar of the form $\sigma^\phi$, and the type
$\tfun{\Gamma}{x \of \sigma}{\tau}$ would then capture all the needed
information, as each type in $\Gamma$ would come with its own
annotation. Instead, we don't embed dependency information in the
types directly, and use annotated context $\Gamma^\Phi$ to carry
equivalent information. This simplification makes it easier to control
the scoping correctness: it is easier to notice that $\Gamma^\Phi$ and
$\Gamma^\Psi$ are contexts ranging over the same domain than if we
wrote $\Gamma$ and $\Gamma'$. It is made possible by two specific
aspects of this simple system:
\begin{itemize}
\item Our intensional information has a very simple structure, only
  a boolean, that does not apply to the types in depth. The
  simplification would not work, for example, in a security type
  system where products could have components of different security
  levels ($\tau^l * \sigma^r$), but the structure of the rules would
  remain the same.

\item In this example, we are interested mostly in intensional
  information on the contexts, rather than the return type of
  a judgment. The general case rather suggests a judgment of the form
  $\Gamma^\Phi \der e : \sigma^\phi$, but with only boolean
  annotations this boils to a judgment
  $\Gamma^\Phi \der e : \sigma^1$, when we are interested in the value
  being type-checked, and a trivial judgment
  $\Gamma^0 \der e : \sigma^0$, used to type-check terms that will not
  be used in the rest of the computation, and which degenerates to
  a check in the simply-typed lambda-calculus. Instead, we define
  a single judgment $\Gamma^\Phi \der e : \sigma$ corresponding to the
  case where $\phi$ is $1$, and use the notation $\phi.\Phi$ to
  nullify the dependency information coming of from $e$ when the outer
  computation does not actually depend on it ($\phi$ is $0$).
\end{itemize}

While dependency annotations make the development easier to follow,
they do not affect the generality of the type theory, as the common
denominator of open closure type systems is more concerned with the
scoping of closure contexts than the structure of the intensional
information itself.
\end{version}

\subsubsection{Maintaining Closure Contexts.}

As pointed out before, the rules \Rule{App-Tmp} and \Rule{Let-Tmp} of the
system above are wrong (hence the ``temporary'' name): the
left-hand-side of the rule \Rule{App-Tmp} assumes that the closure
captures the same environment $\Gamma$ that it is computed in. This
property is initially true in the closure of the rule \Rule{Lam}, but
is not preserved by \Rule{Let-Tmp} (for the body type) or
\Rule{App-Tmp} (for the return type). This means that the intensional
information in a type may become stale, mentioning variables that have
been removed from the context. We will now fix the type system to
never mention unbound variables.

We need a \emph{closure substitution mechanism} to
explain the closure type $\tau_f = \tfun{\Gamma^\Phi,y \of \rho^\chi}{x \of
  \sigma^\phi}{\tau^\psi}$ of a closure $f$ in the smaller environment
$\Gamma$, given dependency information for $y$ in $\Gamma$.  Assume
for example that $y$ was bound in a let binding $\code{let } y = e \ldots$
and that the type $\tau_f$ leaves the scope of $y$.  Then we
have to adapt the type rules to express the following.  ``If $f$
depends on $y$ (at application time) then $f$ depends on the
variables of $\Gamma$ that $e$ depends on.''

\begin{mathparfig}{fig/type-substitution}{Type substitution}
\inferrule[Subst-Context-Nil]
{}
{\Gamma, y \of \rho, \emptyset \rewto{y}{\Psi} \Gamma}
\and
\inferrule[Subst-Context]
{\Gamma, y \of \rho, \Delta \der \sigma
  \rewto{y}{\Psi}
 \Gamma, \Delta' \der \tau}
{\Gamma, y \of \rho, \Delta, x \of \sigma
  \rewto{y}{\Psi}
 \Gamma, \Delta', x \of \tau}
\and
\inferrule[Subst-Atom]
{\Gamma, y \of \rho, \Delta \rewto{y}{\Psi} \Gamma, \Delta'}
{\Gamma, y \of \rho, \Delta \der \alpha \rewto{y}{\Psi} \Gamma, \Delta' \der \alpha}
\and
\inferrule[Subst-Product]
{\Gamma, y \of \rho, \Delta \der \sigma_1
  \rewto{y}{\Psi} \Gamma, \Delta' \der \tau_1\\
 \Gamma, y \of \rho, \Delta \der \sigma_2
 \rewto{y}{\Psi} \Gamma, \Delta' \der \tau_2}
{\Gamma, y \of \rho, \Delta \der \sigma_1 * \sigma_2
 \rewto{y}{\Psi} \Gamma, \Delta' \der \tau_1 * \tau_2}
\and
\inferrule[Subst-Closure-Notin]
{\Gamma_0, \Gamma_1, y \of \rho, \Delta
  \rewto{y}{\Psi} \Gamma_0, \Gamma_1, \Delta'
}
{\Gamma_0, \Gamma_1, y \of \rho, \Delta \der
   \tfun{\Gamma_0^{\Phi}}
        {x \of \sigma_1^{\phi}}
        {\sigma_2}
 \rewto{y}{\Psi} \Gamma_0, \Gamma_1, \Delta' \der
   \tfun{\Gamma_0^{\Phi}}
        {x \of \sigma_1^{\phi}}
        {\sigma_2}
 }
\and
\inferrule[Subst-Closure]
{ \Gamma, y \of \rho, \Delta, \Gamma_1
  \rewto{y}{\Psi} \Gamma, \Delta', \Gamma'_1\\
 \Gamma, y \of \rho, \Delta \vdash \sigma_1
  \rewto{y}{\Psi} \Gamma, \Delta' \der \sigma_1\\
 \Gamma, y \of \rho, \Delta, x \of \sigma_1 \der \sigma_2
   \rewto{y}{\Psi} \Gamma, \Delta', x \of \sigma_1 \der \tau_2
}
{ \!\!\! \Gamma, y \of \rho, \Delta, \Gamma_1 \der
   \tfun{\Gamma^{\Phi_1},y \of \rho^\chi, \Delta^{\Phi_2}}
        {x \of \sigma_1^{\phi}}
        {\sigma_2}
 \rewto{y}{\Psi} \Gamma, \Delta', \Gamma'_1 \der
   \tfun{\Gamma^{\Phi_1+\chi.\Psi}, {\Delta'}^{\Phi_2}}
        {x \of \sigma_1^{\phi}}
        {\tau_2} \!\!\!
 }
\end{mathparfig}

We define in Figure~\ref{fig/type-substitution} the judgment
$\Gamma, y \of \rho, \Delta \der \sigma \rewto{y}{\Psi} \Gamma, \Delta' \der \tau$. Assuming
that the variable $y$ in the context $\Gamma, y\of\rho, \Delta$ was
let-bound to an definition with usage information $\Gamma^\Psi$, this
judgment transforms any type $\sigma$ in this context in a type $\tau$
in a context $\Gamma, \Delta'$ that does not mention $y$ anymore. Note
that $\Delta$ and $\Delta'$ have the same domain, only their
intensional information changed: any mention of $y$ in a closure type
of $\Delta$ was removed in $\Delta'$. Also note that
$\Gamma, y \of \rho, \Delta$ and $\Gamma, \Delta'$, or $\sigma$ and
$\tau$, are not annotated with dependency annotations themselves: this
is only a scoping transformation that depends on the dependency
annotations of $y$ \emph{in the closures} of $\sigma$ and $\Delta$.

As for the scope-checking judgment, we simultaneously define the
substitutions on contexts themselves
$\Gamma, y \of \rho, \Delta \rewto{y}{\Psi} \Gamma, \Delta'$.
There are two rules for substituting a closure type. If the variable
being substituted is not part of the closure type context
(rule \Rule{Subst-Closure-Notin}), this closure type is
unchanged. Otherwise (rule \Rule{Subst-Closure}) the substitution is
performed in the closure type, and the neededness annotation for $y$
is reported to its definition context $\Gamma_0$.

The following lemma verifies that this substitution preserves
well-scoping of contexts and types.

\begin{lemma}[Substitution and scoping]
  \label{lem:substitution_preserves_scoping}
If $\Gamma, y \of \rho, \Delta \der$ and
$\Gamma, y \of \rho, \Delta \rewto{y}{\Psi} \Gamma, \Delta'$
then $\Gamma, \Delta' \der$.
If
$\Gamma, y \of \rho, \Delta \der \sigma$ and
$\Gamma, y \of \rho, \Delta \der \sigma \rewto{y}{\Psi} \Gamma,\Delta' \der \tau$
then $\Gamma,\Delta' \der \tau$.
\end{lemma}


\begin{version}{\Long}
\begin{proof}
  By mutual induction on the judgments
  ${\Gamma, y \of \rho, \Delta \rewto{y}{\Psi} \Gamma, \Delta'}$
  and\\
  ${\Gamma, y \of \rho, \Delta \der \sigma
    \rewto{y}{\Psi} \Gamma, \Delta' \der \tau}$.

  \begin{caselist}
    \nextcase \Rule{Subst-Context-Nil}: using
    \Rule{Scope-Context-Nil}, $\Gamma, x \of \sigma \der$ implies
    $\Gamma \der \sigma$, which in turn implies $\Gamma \der$.

    \nextcase \Rule{Subst-Context}: from our hypothesis
    $\Gamma, y \of \rho, \Delta, x \of \sigma \der$ we deduce
    $\Gamma, y \of \rho, \Delta \der \sigma$. By induction we can
    deduce $\Gamma, \Delta' \der \tau$, which gives context
    well-formedness $\Gamma, \Delta' \der$.

    \nextcase \Rule{Subst-Atom}: direct by \Rule{Scope-Atom} and
    induction hypothesis.
  
    \nextcase \Rule{Subst-Product}: by inversion, the last rule of the
    derivation of $\Gamma, y \of \rho \der (\sigma_1 * \sigma_2)$ is
    \Rule{Scope-Product}, so we can proceed by direct induction on the
    premises of both judgments.

    \nextcase \Rule{Subst-Closure}: 
    Using our induction hypothesis on
    $\Gamma, y \of \rho, \Delta, \Gamma_1
    \rewto{y}{\Psi} \Gamma, \Delta', \Gamma'_1$ we can deduce that
    $\Gamma, \Delta', \Gamma'_1 \der$ and in particular $\Gamma, \Delta' \der$.

    By inversion, the last rule of
    the derivation of
    $\Gamma, y \of \rho, \Delta, \Gamma_1 \der
      \tfun
        {\Gamma^{\Phi_1}, y \of \rho^\chi, \Delta^{\Phi_2}}
        {y : \sigma_1^{\phi_1}}{\sigma_2}$
    is \Rule{Scope-Closure}. Its premises are
    ${\Gamma, y \of \rho, \Delta \der \sigma_1}$ and
    $\Gamma, y \of \rho, \Delta, x \of \sigma_1 \der \sigma_2$, from which we
    deduce by induction hypothesis $\Gamma, \Delta' \der \tau_1$ and
    $\Gamma, \Delta', x \of \tau_1 \der \tau_2$ respectively, allowing to
    deduce that
    $\Gamma, \Delta' \der
      \tfun
        {\Gamma^{\Phi_1+\chi.\Psi}, {\Delta'}^{\Phi_2}}
        {x : \tau_1}{\tau_2}$,
    which allows to conclude by weakening with the well-scoped $\Gamma'_1$.

    \nextcase \Rule{Subst-Closure-Notin}: direct by induction and inversion.
  \end{caselist}
\qed
\end{proof}
\end{version}

We can now give the correct rules for binders:
\begin{mathpar}
\inferrule[Let]
  {\Gamma^\Phidef \der e_1 : \sigma\\
   \Gamma^\Phibody, x \of \sigma^\phi \der e_2 : \tau\\
   \Gamma, x \of \sigma \der \tau \rewto{x}{\Phidef} \Gamma \der \tau'}
  {\Gamma^{\phi.\Phidef+\Phibody} \der
     \clet{x}{e_1}{e_2} : \tau'}
\and
\inferrule[App]
{(\Gamma_0, \Gamma_1)^\Phifun \der
    t :\tfun{\Gamma_0^\Phiclos}{x \of \sigma^\phi}{\tau}\\
  (\Gamma_0, \Gamma_1)^\Phiarg \der u : \sigma\\
  \Gamma_0, \Gamma_1, x \of \sigma \der \tau
    \rewto{x}{\Phiarg} \Gamma_0, \Gamma_1 \der \tau'}
{(\Gamma_0,\Gamma_1)^{\Phifun+\Phiclos+\phi.\Phiarg} \der
    t \app u : \tau'}
\end{mathpar}


\begin{lemma}[Typing respects scoping]
  \label{lem:typing_preserves_scoping}
  If $\Gamma \der t : \sigma$ holds,
  then $\Gamma \der \sigma$ holds.
\end{lemma}

This lemma guarantees that we fixed the problem of stale intensional
information: types appearing in the typing judgment are always
well-scoped.

\begin{version}{\Long}
\begin{proof}
  By induction on the derivation of $\Gamma \der t : \sigma$.
  \begin{caselist}
    \nextcase \Rule{Var}: from the premise
    $\Gamma^0, x:\sigma^1, \Delta^0 \der$ we have
    $\Gamma \der \sigma$.

    \nextcase \Rule{Prod}: direct by induction.

    \nextcase \Rule{Proj}: the induction hypothesis is
    $\Gamma \der \tau_1 * \tau_2$, from which we get
    $\Gamma \der \tau_i$ (for $i \in \{1,2\}$) by inversion.

    \nextcase \Rule{Lam}: the induction hypothesis is
    $\Gamma, x \of \sigma \der \tau$. From this we get
    $\Gamma, x \of \sigma$ and therefore $\Gamma \der \sigma$, which
    allows to conclude with \Rule{Scope-Closure}.

    \Fixpoint{
      \nextcase \Rule{Fix}: the hypothesis implies
      $\Gamma, f \of \tfun{\Gamma^\Psi}{x \of \sigma^\phi}{\tau} \der$,
      which in turn implies
      $\Gamma \der \tfun{\Gamma^\Psi}{x \of \sigma^\phi}{\tau}$.
    }{}

    \nextcase \Rule{App}: Using our induction hypothesis on the first
    premise give us that
    $\Gamma_0, \Gamma_1 \der \tfun{\Gamma_0^{\Phi_{\code{clos}}}}{x \of \sigma}{\tau}$,
    so by inversion $\Gamma_0, \Gamma_1 \der$ and
    $\Gamma_0, x \of \sigma \der \tau$. The latter fact can be weakened into
    $\Gamma_0, \Gamma_1, x \of \sigma \der \tau$, and then combined
    with the last premise
    $\Gamma_0, \Gamma_1, x \of \sigma \der \tau \rewto{x}{\Phi_{\code{arg}}}
     \Gamma_0, \Gamma_1 \der \tau'$
    and Lemma~\ref{lem:substitution_preserves_scoping} to get our goal
    $\Gamma_0, \Gamma_1 \der \tau'$.

    \nextcase \Rule{Let}: reasoning similar to the App case. By
    induction on the middle premise, we have
    $\Gamma, x \of \sigma \der \tau$, combined with the right premise
    $\Gamma, x \of \sigma \der \tau \rewto{x}{\Phidef} \Gamma \der \tau'$
    we get $\Gamma \der \tau'$.

  \end{caselist}
\qed
\end{proof}
\end{version}

It is handy to introduce a convenient derived notation
$\Gamma^\Phi \der \tau \rewto{y}{\Psi} \Gamma'^\Phip \der \tau'$ that is defined below.
This substitution relation does not only remove $y$ from the open closure
types in $\Gamma$, it also updates the dependency annotation on
$\Gamma$ to add the dependency $\Psi$, corresponding to all the
variables that $y$ depended on -- if it is itself marked as needed.
\begin{mathpar} \abovedisplayskip 0em \belowdisplayskip 0em
\inferrule[]
{\Gamma, y \of \rho, \Delta \der \tau \rewto{y}{\Psi} \Gamma, \Delta' \der \tau'}
{\Gamma^{\Phi_1}, y \of \rho^\chi, \Delta^{\Phi_2} \der \tau \rewto{y}{\Psi} \Gamma^{\Phi_1 + \chi.\Psi}, {\Delta'}^{\Phi_2} \der \tau'}
\end{mathpar}

\begin{version}{\Long}
It is interesting to see that substituting $y$ away in
$\Gamma^{\Phi_1}, y \of \rho, \Delta^{\Phi_2}$ changes the annotation
on $\Gamma$, but not its types ($\Gamma$ is unchanged in the output as
its types may not depend on $y$), while it changes the types in
$\Delta$ but not its annotation ($\Phi_2$ is unchanged in the output
as a value for $y$ may only depend on variables from $\Gamma$, not
$\Delta$).
\end{version}

\begin{version}{\Long}
  The following technical results allow us to permute substitutions on
  unrelated variables. They will be used in the typing soundness proof
  of the next section~(Theorem~\ref{thm:sound}).
 
\begin{lemma}[Confluence]
  \label{lem:confluence}
  If $\Gamma_1 \der \tau_1 \rewto{x_a}{\Psi_a} \Gamma_{2a} \der \tau_{2a}$
  and $\Gamma_1 \der \tau_1 \rewto{x_b}{\Psi_b} \Gamma_{2b} \der \tau_{2b}$
  then there exists a unique $\Gamma_3 \der \tau_3$ such that
  \[ \Gamma_{2a} \der \tau_{2a} \rewto{x_b}{(\Psi_b + \Psi_a . \Psi_b(x_a))} \Gamma_3 \der \tau_3
     \qquad \textrm{and} \qquad
     \Gamma_{2b} \der \tau_{2b} \rewto{x_a}{(\Psi_a + \Psi_b . \Psi_a(x_b))} \Gamma_3 \der \tau_3 \]
\end{lemma}

\begin{proof}
  Without loss of generality we can assume that $x_a$ appears before
  $x_b$ in $\Gamma_1$, so in particular $\Psi_a(x_b) = 0$. For any subcontext of the form
  \[ {\Delta_1}^{\Psi_1}, x_a \of {\sigma_a}^{\rho_a}, {\Delta_2}^{\Psi_2}, x_b \of {\sigma_b}^{\rho_b} \],
  assume that substituting $\Psi_a$ for $x_a$ first results in
  \[ {\Delta_1}^{\Psi_1 + \rho_a . \Psi_a}, {\Delta_{2a}}^{\Psi_2}, x_b \of {\sigma_{ba}}^{\rho_b} \],
  while substituting $\Psi_b$ for $x_b$ first results in
  \[ {\Delta_1}^{\Psi_1 + \rho_b . \Psi_b(\Delta_1)}, x_a \of {\sigma_a}^{\rho_a + \rho_b . \Psi_b(x_a)}, {\Delta_2}^{\Psi_2 + \rho_b . \Psi_b(\Delta_2)} \].
  Substituting $\Psi_b + \Psi_a . \Psi_b(x_a)$ for $x_b$ in 
  \[ {\Delta_1}^{\Psi_1 + \rho_a . \Psi_a}, {\Delta_{2a}}^{\Psi_2}, x_b \of {\sigma_{ba}}^{\rho_b} \]
  results in
  \[ {\Delta_1}^{\Psi_1 + \rho_a . \Psi_a + \rho_b . (\Psi_b + \Psi_a . \Psi_b(x_a))(\Delta_1)}, {\Delta_{2a}}^{\Psi_2 + \rho_b . (\Psi_b + \Psi_a . \Psi_b(x_a))(\Delta_{2a})} \]
  which simplifies to
  \[ {\Delta_1}^{\Psi_1 + \rho_a . \Psi_a + \rho_b . \Psi_b(\Delta_1) + \rho_b . \Psi_b(x_a) . \Psi_a}, {\Delta_{2a}}^{\Psi_2 + \rho_b . \Psi_b(\Delta_{2a})} \]
  Substituting $\Psi_a + \Psi_b . \Psi_a(x_b) = \Psi_a$ for $x_a$ in 
  \[ {\Delta_1}^{\Psi_1 + \rho_b . \Psi_b(\Delta_1)}, x_a \of {\sigma_a}^{\rho_a + \rho_b . \Psi_b(x_a)}, {\Delta_2}^{\Psi_2 + \rho_b . \Psi_b(\Delta_2)} \]
  results in
  \[ {\Delta_1}^{\Psi_1 + \rho_b . \Psi_b(\Delta_1) + (\rho_a + \rho_b . \Psi_b(x_a)) . \Psi_a)}, {\Delta_{2a}}^{\Psi_2 + \rho_b . \Psi_b(\Delta_2)} \]
  which simplifies to
  \[ {\Delta_1}^{\Psi_1 + \rho_b . \Psi_b(\Delta_1) + \rho_a . \Psi_a + \rho_b . \Psi_b(x_a) . \Psi_a}, {\Delta_{2a}}^{\Psi_2 + \rho_b . \Psi_b(\Delta_2)} \]

  Given that $\Delta_2$ and $\Delta_{2a}$ have the same domain
  (only different types), the restrictions $\Psi_b(\Delta_2)$ and
  $\Psi_b(\Delta_{2a})$ are equal, allowing to conclude that the two substitutions indeed end in the same sequent
  \[ {\Delta_1}^{\Psi_1 + (\rho_a + \rho_b . \Psi_b(x_a)) . \Psi_a + \rho_b . \Psi_b(\Delta_1)}, {\Delta_{2a}}^{\Psi_2 + \rho_b . \Psi_b(\Delta_2)} \]

  Note that we can make sense, informally, of this resulting sequent. The variable used by this final contexts are
  \begin{itemize}
  \item the variables used of $\Delta_1$ used in the initial judgment ($\Psi_1$)
  \item the variables of $\Delta_2$ (updated in $\Delta_{2a}$ to remove references to the substituted variable $x_a$) used in the initial judgment ($\Psi_2$)
  \item the variables used by $\Psi_b$, if it is used ($\rho_b$ is $1$)
  \item the variables used by $\Psi_a$ if either $x$ was used ($\rho_a$ is $1$), or if $x_b$ is used ($\rho_b$ is $1$) and itself uses $x_a$ ($\Psi_b(x_a)$ is $1$).
  \end{itemize}
  To get this intuition, we considered again the annotations as
  booleans, but note that the equivalence proof was done in a purely
  algebraic manner. It should therefore be preserved in future work
  where the intensional information has a richer structure.

\qed
\end{proof}

\begin{corollary}[Reordering of substitutions]
\label{lem:reordering}
  If $\Psi_a$ and $\Psi_b$ have domain $\Gamma$, and \[ {\Gamma_1}^{\Phi_1} \der \tau_1 \rewto{x_a}{\Psi_a} {\Gamma_2}^{\Phi_2} \der \tau_2 \rewto{x_b}{(\Psi_b + \Psi_b(x_a) . \Psi_a)} {\Gamma_3}^{\Phi_3} \der \tau_3 \]
  then there exists ${\Gammap_2}^{\Phip_2} \der \tau'_2$ such that
  \[ {\Gamma_1}^{\Phi_1} \der \tau_1 \rewto{x_b}{\Psi_b} {\Gammap_2}^{\Phip_2} \der \tau'_2 \rewto{x_a}{(\Psi_a + \Psi_a(x_b) .\Psi_b)} {\Gamma_3}^{\Phi_3} \der \tau_3 \]
\end{corollary}
\end{version}

\begin{version}{\Long}
\subsubsection{On Open Closure Types on the Left of Function Types.}

Note that the \Rule{Subst-Closure} handles the function type on the
left-hand-side of the arrow, $\sigma_1$, is a specific and subtle way:
it must be unchanged by the substitution judgment. Under a slightly
simplified form, the judgment reads:
\begin{mathpar}
  \inferrule
{\Gamma, y \of \rho, \Delta \vdash \sigma_1
  \rewto{y}{\Psi} \Gamma, \Delta' \der \sigma_1\\
 \Gamma, y \of \rho, \Delta, x \of \sigma_1 \der \sigma_2
   \rewto{y}{\Psi} \Gamma, \Delta', x \of \sigma_1 \der \tau_2
}
{\Gamma, y \of \rho, \Delta \der
   \tfun{\Gamma^{\Phi_1},y \of \rho^\chi, \Delta^{\Phi_2}}
        {x \of \sigma_1^{\phi}}
        {\sigma_2}
 \rewto{y}{\Psi} \Gamma, \Delta' \der
   \tfun{\Gamma^{\Phi_1+\chi.\Psi}, {\Delta'}^{\Phi_2}}
        {x \of \sigma_1^{\phi}}
        {\tau_2}
 }
\end{mathpar}

This corresponds to the usual ``change of direction'' on the left of
arrow type. A substitution
$\Gamma, y \of \rho \der \tau \rewto{y}{\Psi} \Gamma \der \tau'$ is
a lossy transformation, as we forget how $y$ is used in $\tau$ and
instead mix its definition information with the rest of the context
information in $\tau'$. Such a loss makes sense for the return type of
a function: we forget information about the return value. But by
contravariance of input argument, we should instead \emph{refine} the
argument types.

But as the gain or loss or precision correspond to variables going out
of scope, such a refinement could only happen in smaller nested
scopes. On the contrary, when going out to a wider scope, the only
possibility is that the closure type does not depend on the particular
variable being substituted (so the type $\sigma$ is preserved,
$\Gamma, y\of\rho, \Delta \der \sigma \rewto{y}{\Psi} \Gamma, \Delta' \der \sigma$)
. If the variable was used, a loss of precision would be
possible: this substitution must be rejected.

\pagebreak
Consider the following example:
\begin{lstlisting}
(* in context $\Gamma$ *)
  let x : int = e_x in 
    let y : bool = e_y in
      let f (g : $\tfun{\Gamma^0, x\of\tyc{int}^1}{z:\tyc{unit^0}}{\tyc{int}}$) : $\tyc{int}$ = g () in
      f ($\lambda$z. x);
      f
\end{lstlisting}

In the environment $\Gamma, x \of \tyc{int}, y \of \tyc{bool}$, the
type of \code{f}'s function argument \code{g} describes a function
whose result depends on \code{x}. We can still express this dependency
when substituting the variable \code{y} away, that is when considering
the type of the expression \code{(let y = ... in let f g = ... in f)}
as a whole: the argument type will still have type
$\tfun{\Gamma^0, x \of \tyc{int}^1}{z:\tyc{unit}}{\tyc{int}}$. However,
this dependency on \code{x} cannot make sense anymore if we remove
\code{x} itself from the context, the substitution does not preserve
this function type. This makes the whole expression
\code{(let x = ... in (let y = ... in let f g = ... in f))} ill-typed,
as \code{x} escapes its scope in the argument function type.

One way to understand this requirement is that there are two parts to
having an analysis be fully ``higher-order''.  Fist, it handles programs
that take functions as input, and second, it handles programs that return functions as
result of computations. Some languages only pass functions as
parameters (this is in particular the case of C with pointers to
global functions), some constructions such as currying fundamentally
rely on function creation with environment capture. Our system
proposes a new way to handle this second part, and is intensionally
simplistic, to the point of being restrictive, on the rest.

In a non-toy language one would want to add subtyping of context
information, that would allow controlled loss of precision to, for
example, create lists of functions with slightly different context
information. Another useful feature would be context information
polymorphism to express functions being parametric with respect to the
context information of their argument. This is intentionally left to
future work.
\end{version}

\section{A Big-Step Operational Semantics}

In this section, we will define an operational semantics for our term
language, and use it to prove the soundness of the type system
(Theorem~\ref{thm:sound}). Our semantics is equivalent to the usual
call-by-value big-step reduction semantics for the lambda-calculus in
the sense that computation happens at the same time. There is however
a notable difference.

Function closures are not built in the same way as they are in classical
big-step semantics. Usually, we have a rule of the form
$ \inferrule{}{V \der \lam x t \redto{} (V, \lam x t)} $ where the
closure for $\lam x t$ is a pair of the value environment $V$
(possibly restricted to its subset appearing in $t$) and the function
code. In contrast, we capture no values at closure creation time in
our semantics:
$ \inferrule{}{V \der \lam x t \redto{} (\emptyset, \lam x t)}$. The
captured values will be added to the closure incrementally, during the
reduction of binding forms that introduced them in the context.

Consider for example the following two derivations; one in the
classic big-step reduction, and the other in our alternative
system.
\begin{mathpar}
\inferrule[Classic-Red-Let]
{x \of v \der x \credto v\\
 x \of v, y \of v \der \lam z y \credto ((x \mapsto v, y \mapsto v), \lam z y)}
{x \of v \der \clet{y}{x}{\lam z y} \credto ((x \mapsto v, y \mapsto v), \lam z y)}
\and
\inferrule[Our-Red-Let]
{x \of v \der x \redto v\\
 x \of v, y \of v \der \lam z y \redto (\fvenv{x,y}, \emptyset, \lam z y)\\
 (\emptyset, \lam z y) \rewto{y}{v} (\fvenv{x}, y \mapsto v, \lam z y)}
{x \of v \der \clet{y}{x}{\lam z y} \redto (\fvenv{x}, y \mapsto v, \lam z y)}
\end{mathpar}

Rather than capturing the whole environment in a closure, we store
none at all at the beginning (merely record their names), and add
values incrementally, just before they get popped from the
environment. This is done by the \emph{value substitution} judgment
$w \rewto{x}{v} w'$ that we will define in this section. The
reason for this choice is that this closely corresponds to our typing
rules, value substitution being a runtime counterpart to substitution
in types $\Gamma \der \sigma \rewto{x}{\Phi} \Gamma' \der \sigma'$;
this common structure is essential to prove of the type soundness
(Theorem~\ref{thm:sound}).

Note that derivations in this modified system and in the original
one are in one-to-one mapping. 
It should not be considered a new dynamic semantics, rather
a reformulation that is convenient for our proofs as it mirrors our
static judgment structure.

\subsubsection{Values and Value Substitution.}

Values are defined as follows.
\[
\begin{array}{l@{~}l@{~}l@{\quad}r}
  \mathtt{Val} \ni v,w & ::= & & \mbox{values} \\
      & \mid & \code{v}_{\alpha} & \mbox{value of atomic type} \\
      & \mid & (v, w) & \mbox{value tuples} \\
      & \mid & (\fvenv{x_j}_\jJ, (x_i \mapsto v_i)_\iI, \lam x t)
               & \mbox{function closures} \\
\Fixpoint{
      & \mid & (\fvenv{x_j}_\jJ, (x_i \mapsto v_i)_\iI, \fix f x t)
               & \mbox{recursive function closures} \\
}{}
\end{array}
\]

The set of variables bound in a closure is split into an ordered
mapping $(x_i \mapsto v_i)_\iI$ for variables that have been
substituted to their value, and a simple list $\fvenv{x_j}_\jJ$ of
variables whose value has not yet been captured. They are both binding
occurrences of variables bound in $t$; $\alpha$-renaming them is
correct as long as $t$ is updated as well.

To formulate our type soundness result, we define a typing judgment on
values $\Gamma \der v : \sigma$ in Figure~\ref{fig/value-typing}. An
intuition for the rule \Rule{Value-Closure} is the
following. Internally, the term $t$ has a dependency $\Gamma^\Phi$ on
the ambient context, but also dependencies $(\tau_i^{\psi_i})$ on the
captured variables. But externally, the type may not mention the
captured variables, so it reports a different dependency
$\Gamma^\Phip$ that corresponds to the internal dependency
$\Gamma^\Phi$, combined with the dependencies $(\Psi_i)$ of the
captured values. Both families $(\psi_i)_\iI$ and $(\Psi_i)_\iI$ are
existentially quantified in this rule.

In the judgment rule, the notation $(x_j:\tau_j)_{j<i}$ is meant to
define the environment of each $(x_i:\tau_i)$ as $\Gamma^\Phi$, plus
all the $(x_j:\tau_j)$ that come before $x_i$ in the typing judgment
$\Gamma^\Phi, (x_i:\tau_i)_\iI, x:\sigma^\phi \der t$. The notation
$\dots \rewto{(x_i)}{(\Psi_i)} \dots$ denotes the sequence of
substitutions for all $(x_i,\Psi_i)$, with the rightmost variable
(introduced last) substituted first: in our dynamic semantics, values
are captured by the closure in the LIFO order in which their binding
variables enter and leave the lexical scope.

\begin{mathparfig}{fig/value-typing}{Value typing}
\inferrule[Value-Atom]
{\Gamma \der}
{\Gamma \der \code{v}_\alpha : \alpha}
\and
\inferrule[Value-Product]
{\Gamma \der v_1 : \tau_1\\
 \Gamma \der v_2 : \tau_2}
{\Gamma \der (v_1,v_2) : \tau_1 * \tau_2}
\and
\inferrule[Value-Closure]
{\Gamma, \Gamma_1 \der\\
  \forall \iI,\uad
    \Gamma, (x_j \of \tau_j)_{j < i} \der v_i : \tau_i\\
 \Gamma^\Phi, (x_i\of\tau_i^{\psi_i})_\iI, x \of \sigma^\phi \der t : \tau\\
 \Gamma^\Phi, (x_i \of \tau_i^{\psi_i})_\iI, x \of \sigma^\phi \der \tau
   \rewto{(x_i)}{(\Psi_i)} \Gamma^\Phip, x \of \sigma^\phi \der \tau'
}
{\Gamma, \Gamma_1 \der (\dom{\Gamma}, (x_i \mapsto v_i)_\iI, \lam x t) :
   \tfun{\Gamma^\Phip}{x \of {\sigma}^\phi}{\tau'}}
\and
\Fixpoint{
  \inferrule[Value-Closure-Fix]
  {\Gamma, \Gamma_1 \der\\
   \forall \iI,\uad
    \Gamma, (x_j \of \tau_j)_{j < i} \der v_i : \tau_i\\
    \Gamma^\Phi, (x_i\of\tau_i^{\psi_i})_\iI,
      f \of (\tfun{\Gamma, (x_i \of \tau_i^{\psi_i})}{x \of \sigma^\phi}{\tau})^\chi,
      x \of \sigma^\phi \der t : \tau\\
    \Gamma^\Phi, (x_i \of \tau_i^{\psi_i})_\iI,       
      x \of \sigma^\phi \der \tau
    \rewto{(x_i)}{(\Psi_i)} \Gamma^\Phip,
      x \of \sigma^\phi \der \tau'
  }
  {\Gamma, \Gamma_1 \der (\dom{\Gamma}, (x_i \mapsto v_i)_\iI, \fix f x t) :
    \tfun{\Gamma^\Phip}{x \of {\sigma}^\phi}{\tau'}}
}{}
\end{mathparfig}

\subsubsection{Substituting Values.}

\begin{version}{\Short}
  The value substitution judgment, define in
  Figure~\ref{fig/value-substitution}, is an operational counterpart
  to the substitution of variables in closures types.
\end{version}
\begin{version}{\Or\Medium\Long}
  In the typing rules, we use the substitution relation
  ${\Gamma, y \of \rho \der \sigma \rewto{y}{\Phi} \Gamma \der \tau}$
  to remove the variable $y$ from the closure types in
  $\sigma$. Correspondingly, in our runtime semantics, we
 \emph{add} the variable $y$ to the vector stored
  in the closure value, by a value substitution judgment
  $w \rewto{y}{v} w'$ (see Figure~\ref{fig/value-substitution})
  when the binding of $y$ is removed from the evaluation context.
\end{version}

\begin{mathparfig}{fig/value-substitution}{Value substitution}
\inferrule[Subst-Value-Atom]
{}
{\code{v}_\alpha \rewto{y}{v} \code{v}_\alpha}
\and
\inferrule[Subst-Value-Product]
{w_1 \rewto{y}{v} w'_1\\
 w_2 \rewto{y}{v} w'_2}
{(w_1, w_2) \rewto{y}{v} (w'_1, w'_2)}
\and
\inferrule[Subst-Value-Closure]
{}
{\left(\fvenv{x_{j_1}, \dots, x_{j_n}, y}, (x_i \mapsto w_i)_\iI, t\right) \rewto{y}{v} 
 \left(\fvenv{x_{j_1}, \dots, x_{j_n}} \push{(y \mapsto v)}(x_i \mapsto w_i)_i, t\right)}
\and
\inferrule[Subst-Value-Closure-Notin]
{y \notin (x_j)_\jJ}
{\left(\fvenv{x_j}_\jJ, (x_i \mapsto w_i)_\iI, t\right) \rewto{y}{v} 
 \left(\fvenv{x_j}_\jJ, (x_i \mapsto w_i)_\iI, t\right)}
\end{mathparfig}

\begin{version}{\Or\Long\Medium}
In the \Rule{Subst-Value-Closure}, the notation
$\push{(y \mapsto v)}(x_i \mapsto w_i)_i$ means that the
binding $y \mapsto v$ is added in first position in the vector of
captured values.  The values $w_i$ were computed in a context depending on
$y$, so they need to appear after the binding $y \mapsto v$ for the value to be type-correct.
\end{version}

\begin{lemma}[Value substitution preserves typing]
  \label{lem:value_substitution_preserves_typing}
  If
  $(\Gamma \der v : \rho)$,
  $(\Gamma, y \of \rho \der w : \sigma)$,
  $(\Gamma, y \of \rho \der \sigma
    \rewto{y}{\Psi} \Gamma \der \tau)$ and
  $(w \rewto{y}{v} w')$ hold,
  then $(\Gamma \der w' : \tau)$ holds.
\end{lemma}

\begin{version}{\Or\Medium\Long}
  Note that this theorem is restricted to substitutions of the
  rightmost variable of the context. While substitution in types needs
  to operate under binders (in rule \Rule{Subst-Closure}), value
  substitution is a runtime operation that will only be used by our
  (weak) reduction relation on the last introduced variable.
\end{version}

\begin{version}{\Long}
\begin{proof}
  By induction on the value typing judgment
  $\Gamma, y \of \rho \der w : \sigma$.

  \begin{caselist}
    \nextcase \Rule{Value-Atom}: by inversion we know that the last
    rule of
    ${\Gamma, y \of \rho \der \alpha \rewto{y}{\Psi} \Gamma \der \tau}$
    is \Rule{Subst-Atom}. From the premise $\Gamma, y \of \rho \der$
    and $\Gamma, y \of \rho \rewto{y}{\Psi} \Gamma$ we can deduce from
    Lemma~\ref{lem:substitution_preserves_scoping} that $\Gamma \der$,
    which allows to conclude $\Gamma \der \code{v}_\alpha : \alpha$.

    \nextcase \Rule{Value-Product}: by inversion, the last rules of
    ${\Gamma, y \of \rho \der (\sigma_1 * \sigma_2)
        \rewto{y}{\Psi} \Gamma \der \tau}$
    and $w \rewto{y}{v} w'$ are respectively \Rule{Subst-Product} and
    \Rule{Subst-Value-Product}, so by induction hypothesis we have
    $\Gamma \der w'_1 : \tau_1$ and
    $\Gamma \der w'_2 : \tau_2$, which allows us to conclude
    $\Gamma \der (w'_1, w'_2) : \tau_1 * \tau_2$.

    \nextcase \Rule{Value-Closure}. By inversion we know that the last
    substitution rules applied are either \Rule{Subst-Closure} and
    \Rule{Subst-Value-Closure}, or \Rule{Subst-Closure-Notin} and
    \Rule{Subst-Value-Closure-Notin}, depending on whether the
    substituted variable is part of the closure context.

    In the latter case, both the types and the judgments are
    unchanged, so the result is immediate. If the substituted value is
    part of the closure context, the rules of the involved judgments
    are
    \begin{mathpar}
\inferrule
{\Gamma, y \of \rho \vdash \sigma_1
  \rewto{y}{\Psi} \Gamma \der \sigma_1\\
 \Gamma, y \of \rho, x \of \sigma_1 \der \sigma_2
   \rewto{y}{\Psi} \Gamma, x \of \sigma_1 \der \sigma'_2
}
{\Gamma, y \of \rho \der
   \tfun{\Gamma^{\Phi_1},y \of \rho^\chi}
        {x \of \sigma_1^{\phi}}
        {\sigma_2}
 \rewto{y}{\Psi} \Gamma \der
   \tfun{\Gamma^{\Phi_1+\chi.\Psi}}
        {x \of {\sigma_1}^{\phi}}
        {\sigma'_2}
 }
\and
\inferrule
{\forall \iI,\uad
    \Gamma, y \of \rho, (x_j \of \tau_j)_{j < i} \der w_i : \tau_i\\
 \Gamma^{\Phi''_1},y \of \rho^{\chi''},
   (x_i\of\tau_i^{\psi_i})_\iI, x \of {\sigma_1}^\phi \der t : \sigma''_2\\
 \Gamma^{\Phi''_1},y \of \rho^{\chi''},
   (x_i \of \tau_i^{\psi_i})_\iI, x \of {\sigma_1}^\phi \der \sigma''_2
   \rewto{(x_i)}{(\Psi_i)}
     \Gamma^{\Phi_1},y \of \rho^{\chi},
       x \of {\sigma_1}^\phi \der \sigma_2
}
{\Gamma, y \of \rho \der
  (\dom \Gamma, (x_i \mapsto w_i)_\iI, \lam x t) :
   \tfun
     {\Gamma^{\Phi_1},y \of \rho^{\chi}}
     {x \of {\sigma_1}^\phi}{\sigma_2}}
\and
\inferrule
{}
{\left(\fvenv{\dom \Gamma, y},
       (x_i \mapsto w_i)_\iI,
       t\right)
 \rewto{y}{v} 
 \left(\dom \Gamma
       \push{(y \mapsto v)}(x_i \mapsto w_i)_i,
       t\right)}  
    \end{mathpar}

    We can reach our goal by using the following inference rule:
    \begin{mathpar}
\inferrule
{\forall \iI,\uad
    \Gamma, y \of \rho, (x_j \of \tau_j)_{j < i} \der w_i : \tau_i\\
    \Gamma^{\Phi''_1},
    y : \rho^{\chi''}, (x_i\of{\tau_i}^{\psi_i})_\iI, x \of {\sigma_1}^\phi
   \der  t : \sigma''_2\\
   \Gamma^{\Phi''_1},
     y \of \rho^{\chi''}, (x_i \of \tau_i^{\psi_i})_\iI, x \of {\sigma_1}^\phi
     \der \sigma''_2
   \rewto{\push{y}{(x_i)_i}}{\push{\Psi}{(\Psi_i)_i}}
     \Gamma^{\Phi_1+\chi.\Psi},
       x \of {\sigma_1}^\phi \der \sigma'_2
 }
{\Gamma \der
  \left(\dom \Gamma, \push{(y \mapsto v)}{(x_i \mapsto w_i)_\iI},
    \lam x t\right) :
  \tfun
    {\Gamma^{\Phi_1+\chi.\Psi}}
    {x \of {\sigma_1}^\phi}{\sigma'_2}}
\end{mathpar}



The typing assumptions all match those of our premise. The reduction
assumption is simply the composition of the reductions of the premises: \[
  \begin{array}{ll}
    &
      \Gamma^{\Phi''_1}, y \of \rho^{\chi''}, (x_i \of \tau_i^{\psi_i})_\iI,
        x \of {\sigma_1}^\phi \der \sigma''_2 \\
    \rewto{(x_i)}{(\Psi_i)} &
      \Gamma^{\Phi_1}, y \of \rho^{\chi}, x \of {\sigma_1}^\phi \der \sigma_2 \\
    \rewto{y}{\Psi} &
      \Gamma^{\Phi_1+\chi.\Psi}, x \of {\sigma_1}^\phi \der \sigma'_2 \\
  \end{array}
  \]
  \end{caselist}
\qed
\end{proof}
\end{version}


\subsubsection{The Big-Step Reduction Relation.}

We are now equipped to define in Figure~\ref{fig/reduction} the
big-step reduction relation on well-typed terms $V \der e \redto v$,
where $V$ is a mapping from the variables to values that is assumed to
contain at least all the free variables of $e$. The notation
$w \Rewto{V_2} w'$ denotes the sequence of substitutions for each
(variable, value) pair in $V_2$, from the last one introduced in the
context to the first; the intermediate values are unnamed and
existentially quantified.

\begin{mathparfig}{fig/reduction}{Big-step reduction rules}
\inferrule[Red-Var]
{}
{V \der x \redto V(x)}
\and
\inferrule[\hspace{-0.2cm}Red-Lam]
{}
{\hspace{-0.2cm} V \der \lam x t \redto \! (\dom V, \emptyset, \lam x t) \hspace{-0.9cm}}
\and
\Fixpoint{
  \inferrule[Red-Lam-Fix]
  {}
  {V \der \fix f x t \redto (\dom V, \emptyset, \fix f x t)}
}{}
\and
\inferrule[Red-Pair]
{\! V \der e_1 \redto \! v_1\\
 \! V \der e_2 \redto \! v_2}
{V \der (e_1, e_2) \redto (v_1, v_2)}
\and
\inferrule[Red-Proj]
{V \der e \redto (v_1, v_2)}
{V \der \pi_i(e) \redto v_i}
\and
\inferrule[Red-Let]
  {V \der e_1 \redto v_1\\
   V\push{(x \mapsto v_1)} \der e_2 \redto v_2\\
   v_2 \rewto{x}{v_1} v'_2
  }
  {V \der \clet{x}{e_1}{e_2} \redto v'_2}
\and
\inferrule[Red-App]
{V, V_1 \der t \redto (\dom V, V_2, \lam y t')\\
 V, V_1 \der u \redto v_\code{arg}\\
 V, V_1, V_2, y \mapsto v_\code{arg}
   \der t' \redto w\\
 w \rewto{y}{v_\code{arg}} w' \Rewto{V_2} w''
}
{V, V_1 \der t \app u \redto w''}
%
%
\Fixpoint{
  \inferrule[Red-App-Fix]
  {V, V_1 \der t \redto (\dom V, (x_i \mapsto v_i)_\iI, \fix f y {t'})\\
   V, V_1 \der u \redto v_\code{arg}\\
   V\push{(x_i \mapsto v_i)_i}
    \push{(f \mapsto \fix f y {t'})}
    \push{(y \mapsto v_{\code{arg}})}
    \der t' \redto w
  }
  {V, V_1 \der t \app u \redto w}
}{}
\end{mathparfig}

We write $V : \Gamma \der$ if the context valuation $V$, mapping free
variables to values, is well-typed according to the context
$\Gamma$. \Short{The definition of this judgment is given in the full
  version.}{}
\begin{version}{\Or\Medium\Long}
  \begin{mathpar} \noabovedisplayskip \nobelowdisplayskip
    \inferrule[Value-Env-Empty] {} {\emptyset : \emptyset \der}
\and
    \inferrule[Value-Env]
    {V : \Gamma \der\\
      \Gamma \der v : \sigma} {V\push{x \mapsto v} : \Gamma, x \of \sigma \der}
  \end{mathpar}
\end{version}

\begin{theorem}[Type soundness]\label{thm:sound}
  If
  $\Gamma^{\Phi} \der t : \sigma$,
  $V : \Gamma \der$ and
  $V \der t \redto v$ 
  then $\Gamma \der v : \sigma$.
\end{theorem}

\begin{version}{\Or\Medium\Long}
\begin{proof}
  By induction on the reduction derivation $V \der t \redto v$.

  \begin{caselist}
    \nextcase \Rule{Red-Var}: From $V : \Gamma \der$ we have
    $\Gamma \der V(x) : \sigma$.

    \nextcase \Rule{Red-Lam}\Fixpoint{, \Rule{Red-Lam-Fix}}{}: in the
    degenerate case where there are no captured values, the value
    typing rule \Rule{Value-Closure} for
    $\tfun{\Gamma^\Phi}{x \of \sigma^\phi}{\tau}$ has as only premise
    $\Gamma^\Phi, x \of \sigma^\phi \der \tau$, which is precisely our
    typing assumption.

    \nextcase \Rule{Red-Pair}, \Rule{Red-Proj}: direct by
    induction.

    \nextcase \Rule{Red-Let}:
The involved derivations are the following:
\begin{mathpar}
  \inferrule
  {V \der e_1 \redto v_1\\
   V\push{x \mapsto v_1} \der e_2 \redto v_2\\
   v_2 \rewto{x}{v_1} v'_2}
  {V \der \clet{x}{e_1}{e_2} \redto v'_2}
\and
  \inferrule
  {\Gamma^\Phidef \der e_1 : \sigma\\
    \Gamma^\Phibody, x \of \sigma^\phi \der e_2 : \tau\\
    \Gamma, x \of \sigma \der \tau \rewto{x}{\Phidef} \Gamma \der \tau'}
  {\Gamma^{\phi.\Phidef+\Phibody} \der
    \clet{x}{e_1}{e_2} : \tau'}
\end{mathpar}
By induction hypothesis we have $\Gamma \der v_1 : \sigma$, from which
we deduce that $V\push{x \mapsto v_1} : \Gamma, x\of\sigma \der$. This
allows us to use induction again to deduce that
$\Gamma, x \of \sigma \der v_2 : \tau$. Finally, preservation of value
typing by value substitution
(Lemma~\ref{lem:value_substitution_preserves_typing}) allows to
conclude that the remaining goal $\Gamma \der v'_2 : \tau'$ holds.

\nextcase \Rule{Red-App}: the involved derivations are the following.
\begin{mathpar}
\inferrule
{V, V_1 \der t \redto (\dom V, (x_i \mapsto v_i)_\iI, \lam y t')\\
 V, V_1 \der u \redto v_\code{arg}\\
 V, V_1 \push{(x_i \mapsto v_i)_i}\push{y \mapsto v_\code{arg}}
   \der t' \redto w\\
 w \rewto{y}{v_\code{arg}} w' \rewto{(x_i)}{(v_i)} w''
}
{V, V_1 \der t \app u \redto w''}
\and
\inferrule
{(\Gamma, \Gamma_1)^\Phifun \der
    t :\tfun{\Gamma^\Phiclos}{y \of \sigma^\phi}{\tau}\\
 (\Gamma, \Gamma_1)^\Phiarg \der u : \sigma\\
 \Gamma, y \of \sigma \der \tau \rewto{y}{\Phiarg} \Gamma \der \tau'}
{(\Gamma, \Gamma_1)^{\Phifun+\Phiclos+\phi.\Phiarg} \der
    t \app u : \tau'}
\end{mathpar}
By induction we have that
$\Gamma, \Gamma_1 \der v_\code{arg} : \sigma$ and
$\Gamma \der ((x_i \mapsto v_i)_i, \lam y {t'}) :
  \tfun{\Gamma^\Phiclos}{y \of \sigma^\phi}{\tau}$.

By inversion, we know that the derivation for this value typing
judgment is of the form \[
\inferrule
{\Gamma, \Gamma_1 \der\\
 \forall \iI,\uad
    \Gamma, (x_j \of \tau_j)_{j < i} \der v_i : \tau_i\\
 \Gamma^\Phi, (x_i\of\tau_i^{\psi_i})_\iI, y \of {\sigma}^\phi \der t' : \tau''\\
 \Gamma^\Phi, (x_i \of \tau_i^{\psi_i})_\iI, y \of {\sigma}^\phi \der \tau''
   \rewto{(x_i)}{(\Psi_i)} \Gamma^{\Phiclos}, y \of \sigma^\phi \der \tau
}
{\Gamma, \Gamma_1 \der (\dom \Gamma, (x_i \mapsto v_i)_\iI, \lam y t') :
   \tfun{\Gamma^{\Phiclos}}{y \of \sigma^\phi}{\tau}}
\]

From our result $\Gamma, \Gamma_1 \der v_{\code{arg}}$ and the premises $\Gamma,
(x_j\of\tau_j)_{j<i} \der v_i : \tau_i$ we deduce that the application
valuation is well-typed with respect to the application environment:
$V, V_1\push{(x_i \mapsto v_i)_i}\push{y \mapsto v_\code{arg}} : \Gamma, \Gamma_1, (x_i\of\tau_i),\sigma \der$.
It is used in the premise of the body computation judgment, so by induction we get that
$\Gamma, \Gamma_1, (x_i\of\tau_i), y\of\sigma \der w : \tau''$.

From there, we wish to use preservation of value typing on the chain
of value substitutions $
  w \rewto{y}{v_{\code{arg}}} w' \rewto{(x_i)}{(v_i)} w''
$ to conclude that $\Gamma \der w'' : \tau'$. 
However, the type substitutions of our premises are in the reverse
order:
\[
\Gamma^\Phi, (x_i \of \tau_i^{\psi_i})_\iI, y \of {\sigma}^\phi \der \tau''
   \rewto{(x_i)}{(\Psi_i)} \Gamma^{\Phiclos}, y \of \sigma^\phi \der \tau
   \rewto{y}{\Phiarg} \Gamma^{\Phiclos + \phi.\Phiarg} \der \tau'
\]
we therefore need to use our Reordering Lemma~(\ref{lem:reordering})
to get them in the right order --- note that the annotations $(\Psi_i)_\iI$ and
$\Phiarg$ are not changed as they are independent from each other: for
any $i$, we have $\Psi_i(y) = \Phiarg(x_i) = 0$.
\[
\Gamma^\Phi, (x_i \of \tau_i^{\psi_i})_\iI, y \of {\sigma}^\phi \der \tau''
   \rewto{y}{\Phiarg} \Gamma^{\Phipclos}, (x_i \of \tau_i^{\psi_i})_\iI \der \tau
   \rewto{(x_i)}{(\Psi_i)} \Gamma^{\Phiclos + \phi.\Phiarg} \der \tau'
\]

\end{caselist}
\qed
\end{proof}
\end{version}

\begin{mathparfig}{fig/classic-reduction}{Classic big-step reduction rules}
\begin{version}{\Or\Medium\Long}
\inferrule[Classic-Red-Var]
{}
{W \der x \credto W(x)}
\end{version}
\and
\inferrule[Classic-Red-Lam]
{}
{W \der \lam x t \credto (W, \lam x t)}

\begin{version}{\Or\Medium\Long}
\Fixpoint{
  \inferrule[Classic-Red-Lam-Fix]
  {}
  {W \der \fix f x t \credto (W, \fix f x t)}
}{}
\and
\inferrule[Classic-Red-Pair]
{W \der e_1 \credto w_1\\
 W \der e_2 \credto w_2}
{W \der (e_1, e_2) \credto (w_1, w_2)}
\and
\inferrule[Classic-Red-Proj]
{W \der e \credto (w_1, w_2)}
{W \der \pi_i(e) \credto w_i}
\end{version}
\and
\inferrule[Classic-Red-Let]
  {W \der e_1 \credto w_1\\
   W\push{x \mapsto w_1} \der e_2 \credto w_2
  }
  {W \der \clet{x}{e_1}{e_2} \credto w_2}
\begin{version}{\Or\Medium\Long}
\\
\end{version}
\and
\inferrule[Classic-Red-App]
{W \der t \credto (W', \lam y t')\\
 W \der u \credto w_\code{arg}\\
 W'\push{y \mapsto w_\code{arg}} \der t' \credto w
}
{W \der t \app u \credto w}
\and
\begin{version}{\Or\Medium\Long}
\Fixpoint{
  \inferrule[Classic-Red-App-Fix]
  {W \der t \credto (W', \fix f y {t'})\\
    W \der u \credto w_\code{arg}\\
    W'\push{f \mapsto (W', \fix f y {t'})}
      \push{y \mapsto w_{\code{arg}}}
    \der t' \credto w
  }
  {W \der t \app u \credto w}
}{}
\end{version}
\end{mathparfig}

Finally, we recall the usual big-step semantics for the call-by-value
calculus with environments, in Figure~\ref{fig/classic-reduction}, and
state its equivalence with our utilitarian semantics. Due to space
restriction we will only mention the rules that differ, and elide the
equivalence proof, but the long version contains all the details.

There is a close correspondence between judgments of both semantics,
but as the value differ slightly, in the general cases the value
bindings of the environment will also differ. We state the theorem
only for closed terms, but the proof will proceed by induction on
a stronger induction hypothesis using an equivalence between non-empty
contexts.

\begin{theorem}[Semantic equivalence]\label{thm:sem-equiv}
  Our reduction relation is equivalent with the classic one on closed
  terms: $\emptyset \der t \redto v$ holds if and only if
  $\emptyset \der t \credto v$ also holds.
\end{theorem}

  To formulate our induction hypothesis, we define the equivalence
  judgment $V \der v = W \cder w$; on each side of the equal sign
  there is a context and a value, the right-hand side being considered
  in the classical semantics.

  \begin{mathparfig}{fig/equiv}{Equivalence of semantic judgements.}
    \inferrule
      {}
      {\emptyset \der{} = \emptyset \cder{} }
\and
    \inferrule
      { V \der{} \! = W \cder{}
      \\ V \der v = W \cder w  }
      {V\push{x \mapsto v} \der{} = W\push{x \mapsto w} \cder{}}
\and
    \inferrule
      {V \der{} = W \cder{}}
      {V \der \code{v}_\alpha = W \cder \code{v}_\alpha }
\and
    \inferrule
      {V \der v_1 = W \cder w_1\\ \!\!\!
       V \der v_2 = W \cder w_2}
      {V \der (v_1, v_2) = W \cder (w_1, w_2)}
\and
    \inferrule
      {V \der{} = W \cder{}\\ 
       V\push{x_i \mapsto v_i} \der{} = W' \cder{}}
      {V \der ((x_i \mapsto v_i)_\iI, \lam x t) =
       W \cder (W', \lam x t)}
  \end{mathparfig}

The stronger version of the theorem becomes the following: if
$V \der{} = W \cder{}$ and $V \der t \redto v$ and
$W \der t \credto w$, then $V \der v = W \cder w$.

\begin{version}{\Long}  
  As for subject reduction, we first need to prove that value
  substitution preserves value equivalence.
  
  \begin{lemma}[Value substitution preserves value equivalence]
    If $V\push{y \mapsto v_0} \der v = W \cder w$,
    $V \der v_0 = W \cder w_0$,
    and $v \rewto{y}{v_0} v'$
    then $V \der v' = W \cder w$.
  \end{lemma}
  \begin{proof}
    \nextcase \Rule{Subst-Value-Atom}: direct.
    
    \nextcase \Rule{Subst-Value-Product}: direct induction.

    \nextcase \Rule{Subst-Value-Closure}: We have
    $((x_i \mapsto v_i)_\iI, \lam x t) \rewto{y}{v_0} (\push{y \mapsto v_0}(x_i \mapsto v_i)_\iI, \lam x t)$
    and
    $ V\push{y \mapsto v_0} \der ((x_i \mapsto v_i)_\iI, \lam x t) = W \der (W', \lam x t)$. The
    latter implies
    $V\push{y \mapsto v_0}\push{x_i \mapsto v_i} \der{} = W' \cder{}$
    which which in turn implies our goal
    $V \der (\push{y \mapsto v_0})$ -- with the additional premise
    $V \der{} = W \cder{}$ from our hypothesis
    $V \der v_0 = W \cder w_0$.
  \qed
  \end{proof}

  We can now prove the theorem proper:
  \begin{proof}
    \nextcase \Rule{Red-Var}, \Rule{Classic-Red-Var}: we check by
    direct induction on the contexts that if $V \der{} = W \cder{}$
    holds, then $V \der V(x) = W \cder W(x)$.

    \nextcase \Rule{Red-Lam}, \Rule{Classic-Red-Lam}: The
    correspondence between $V \der (\emptyset, \lam x t)$ and
    $W \cder (W, \lam x t)$ under assumption $V \der{} = W \cder{}$
    is direct from the inference rule
    \[
    \inferrule
      {V\push{\emptyset} \der{} = W \cder{}}
      {V \der (\emptyset, \lam x t) =
       W \cder (W, \lam x t)}
     \]

    \nextcase \Rule{Red-Pair}, \Rule{Classic-Red-Pair} and
    \Rule{Red-Proj}, \Rule{Classic-Red-Proj}: direct by induction.

    \nextcase \Rule{Red-Let}, \Rule{Classic-Red-Let}: By induction
    hypothesis on the $e_1$ premise, we deduce that
    $V \der v_1 = W \cder w_1$, hence
    $V\push{x \mapsto v_1} \der{} = W\push{x \mapsto w_1} \cder{}$.
    This lets us deduce, again by induction hypothesis, that
    $V\push{x \mapsto v_1} \der v_2 = W\push{x \mapsto w_1} \cder w_2$.
    As value substitution preserves value equivalence, we can deduce
    from $v_2 \rewto{x}{v_1} v'_2$ that
    our goal $V \der v'_2 = W \cder w_2$ holds.

    \nextcase \Rule{Red-App}, \Rule{Classic-Red-App}: By induction
    hypothesis we have that
    $V \der v_{\code{arg}} = W \der w_{\code{arg}}$ and
    $V \der ((x_i \mapsto v_i)_\iI, \lam y {t'}) = W \cder (W', \lam y {t'})$.
    By inversion on the value equivalence judgment of the latter we
    know that  $V\push{x_i \mapsto v_i} \der{} = W' \cder{}$. Combined
    with the former value equivalence, this gives us
    $V\push{x_i \mapsto v_i}\push{y \mapsto v_{\code{arg}}} \der{} =
     W'\push{y \mapsto w_{\code{arg}}} \cder{}$,
    so by induction hypothesis and preservation of value equivalence
    by reduction we can deduce our goal.
  \qed
  \end{proof}
\end{version}

\section{Dependency Information as Non-Interference}

We can formulate our dependency information as
a \emph{non-interference} property. Two valuations $V$ and $V'$ are
$\Phi$-equivalent, noted $V =_\Phi V'$, if they agree on all variables
on which they depend according to $\Phi$. We say that $e$ respects
non-interference for $\Phi$ when, whenever $V \der e \redto v$ holds,
then for any $V'$ such that $V =_\Phi V'$ we have that
$V' \der e \redto v$ also holds. This corresponds to the
information-flow security idea that variables marked $1$ are
low-security, while variables marked $0$ are high-security and should
not influence the output result.

This non-interference statement requires that the two
evaluations of $e$ return the same value $v$. This raises the question
of what is the right notion of equality on values. Values of atomic
types have a well-defined equality, but picking the right notion of
equality for function types is more difficult. While we can state
a non-interference result on atomic values only, the inductive
subcases would need to handle higher-order cases as well.

Syntactic equality (even modulo $\alpha$-equivalence) is not the right
notion of equality for closure values. Consider the following example:
$x\of\tau^0 \der \clet{y}{x}{\lam z z} :
\tfun{x\of\tau^0}{z:\sigma^1}{\sigma}$. This term contains an occurrence of the
variable $x$, but its result does not depend on it. However,
evaluating it under two different contexts $x \of v$ and $x \of v'$,
with $v \neq v'$, returns distinct closures: $(x \mapsto v, \lam z z)$ on one
hand, and $(x \mapsto v', \lam z z)$ on the other. These closures are
not structurally equal, but their difference is not essential since
they are indistinguishable in any context. Logical relations are the
common technique to ignore those internal differences and get a more
observational equality on functional values. They involve, however, a
fair amount of metatheoretical effort \Fixpoint{(in particular in
  presence of non-terminating fixpoints)}{} that we would like to
avoid.

Consider a different example:
$x\of\tau^0 \der \lam y x : \tfun{x \of \tau^1}{y \of \sigma^0}{\tau}$. Again,
we could use two contexts $x \of v$ and $x \of v'$ with $v \neq v'$, and we
would get as a result two closures: 
$x \of v \der \lam y x \redto (x \mapsto v, \lam y x)$ and
$x \of v' \der \lam y x \redto (x \mapsto v', \lam y x)$.
Interestingly, these two closures are \emph{not} equivalent under all
contexts: any context applying the function will be able to observe
the different results. However, our notion of interference requires
that they can be considered equal.  This is motivated by real-world programming
languages that only output a pointer to a closure in a program that
returns a function.

While the aforementioned closures are not equal in any context, they are in
fact equivalent from the point of view of the particular dependency
annotation for which we study non-interference, namely
$x \of \tau^0$. To observe the difference between those closures, we
would need to apply the closure of type
$\tfun{x\of\tau^1}{y:\sigma}{\tau}$, so would be in the different context
$x \of \tau^1$.

This insight leads us to our formulation of value equivalence in
Figure~\ref{fig/value-equiv}. Instead of being as modular and general
as a logical-relation definition, we fix a \emph{global dependency}
$\Phi_0$ that restricts which terms can be used to differentiate
values.
\begin{mathparfig}{fig/value-equiv}{Value equivalence}
  \inferrule[Equiv-Atom]
     {}
     {\Gamma \der \code{v}_\alpha =_{\Phi_0} \code{v}_\alpha : \alpha}
\and
  \inferrule[Equiv-Pair]
     {\Gamma \der v_1 =_{\Phi_0} v'_1 : \sigma_1\\
      \Gamma \der v_2 =_{\Phi_0} v'_2 : \sigma_2}
     {\Gamma \der (v_1,v_2) =_{\Phi_0} (v'_1,v'_2) : \sigma_1 * \sigma_2}
\and
  \inferrule[Equiv-Closure]
     {\forall \iI,\uad
       \Gamma, (x_j \of \tau_j)_{j < i} \der v_i : \tau_i\\
      \Gamma^\Phi, (x_i\of\tau_i^{\psi_i})_\iI, x \of \sigma^\phi \der t : \tau\\
      \Gamma^\Phi, (x_i \of \tau_i^{\psi_i})_\iI, x \of \sigma^\phi \der \tau
      \rewto{(x_i)}{(\Psi_i)} \Gamma^\Phip, x \of \sigma^\phi \der \tau'\\
      \forall \iI, \Psi_i \subseteq \Phi_0 \implies v_i =_{\Phi_0} v'_i}
     {\Gamma \der ((x_i \mapsto v_i)_\iI ,\lam y t)
       =_{\Phi_0}
       ((x_i \mapsto v'_i)_\iI,\lam y t) :
       \tfun{\Gamma^\Phip}{x \of \sigma}{\tau'}}
\end{mathparfig}

Our notion of value equivalence,
$\Gamma \der v =_{\Phi_0} v' : \sigma$ is typed and includes
structural equality. In the rule \Rule{Equiv-Closure}, we check that
the two closures values are well-typed, and only compare captured
values whose dependencies are included in those of the global context
$\Phi_0$, as we know that the others will not be used. 
This equality is tailored to the need of
the non-interference result, which only compares values resulting from
the evaluation of the same subterm -- in distinct contexts.

\begin{theorem}[Non-interference]\label{thm:non-inter}
  If $\Gamma^{\Phi_0} \der e : \sigma$ holds, then for any contexts
  $V,V'$ such that $V =_{\Phi_0} V'$ and values $v,v'$ such that
  $V \der e \redto v$ and $V' \der e \redto v'$, we have
  $\Gamma \der v =_{\Phi_0} v' : \sigma$. In particular, if $\sigma$
  is an atomic type, then $v = v'$ holds.
\end{theorem}

\begin{version}{\Or\Medium\Long}
\begin{proof}
  We will proceed by simultaneous induction on the typing derivation
  $\Gamma \der e : \sigma$ and reduction derivation
  $V \der e \redto v$ and . Note that we use a different induction
  hypothesis: for any subderivations $\Gamma^\Phi \der e : \sigma$ and
  $V \der t \redto v$, we will prove that for any $V'$ that agrees
  with $V$ on $\Phi$ modulo $\Phi_0$-equivalence
  ($\forall x, V(x) =_{\Phi_0} V'(x)$, which we still note
  $V =_\Phi V'$), and with $V' \der t \redto v'$, we have
  $v =_{\Phi_0} v'$.

  We will omit the contexts and types $\Gamma, \sigma$ of a value
  equivalence $\Gamma \der v =_{\Phi_0} v' : \sigma$ when they are
  clear from the context.

  \begin{caselist}
    \nextcase \Rule{Red-Var}: from $V =_{0,x:1,0} V'$ we have
    $V(x) =_{\Phi_0} V'(x)$.

    \nextcase \Rule{Red-Lam}\Fixpoint{, \Rule{Red-Lam-Fix}}{}
    : the returned value
    does not depend on the environment.

    \nextcase \Rule{Red-Pair}, \Rule{Red-Proj}: direct by induction.

    \nextcase \Rule{Red-Let}: the involved derivations are the form
    \begin{mathpar}
    \inferrule
    {\Gamma^\Phidef \der e_1 : \sigma\\
     \Gamma^\Phibody, x\of\sigma^\phi \der e_2 : \tau}
    {\Gamma^{\Phibody + \phi.\Phidef} \der \clet{x}{e_1}{e_2} : \tau}
\and
    \inferrule
    {V \der e_1 \redto v_1\\
      V\push{x \mapsto v_1} \der e_2 \redto v_2\\
      v_2 \rewto{x}{v_1} v_3
    }
    {V \der \clet{x}{e_1}{e_2} \redto v_3}
\and
    \inferrule
    {V' \der e_1 \redto v'_1\\
      V'\push{x \mapsto v'_1} \der e_2 \redto v'_2\\
      v'_2 \rewto{x}{v'_1} v'_3
    }
    {V' \der \clet{x}{e_1}{e_2} \redto v'_3}
    \end{mathpar}

    The context equivalence $V =_{\Phibody + \phi.\Phidef} V'$ implies
    the weaker equivalence $V =_\Phibody V'$, from which we can deduce
    $V\push{x \mapsto v_1} =_{(\Phibody,\phi)} V'\push{x \mapsto v'_1}$
    regardless of the value of $\phi$. Indeed, if $\phi$ is $0$ this
    is direct, and if $\phi$ is $1$ we get $v_1 =_{\Phi_0} v'_1$ by
    induction hypothesis. From this equality between contexts we get
    $v_2 =_{\Phi_0} v'_2$ by induction hypothesis.
   
    We then reason by case distinction on the property
    $\Phidef \subseteq \Phi_0$. If it holds, then again
    $v_1 =_{\Phi_0} v'_1$ by induction, so substituting $v_1,v'_1$ in
    the closures of $v_2,v'_2$ will preserve $\Phi_0$-equivalence. And
    if it does not, those values $v_1,v'_1$ captured in the closures
    of $v_2,v'_2$ will not be tested for $\Phi_0$-equivalence. In any
    case, we have $v_3 =_{\Phi_0} v'_3$.

    \nextcase \Rule{Red-App}: the proof for the application case uses
    the same mechanisms as for the \Rule{Red-Let} case but is more
    tedious because of the repeated application and substitution of
    the closed-over values. To simplify notations, we will handle the
    case of a single captured value $x \mapsto v$. The involved
    derivations are the following:

\begin{mathpar}      
\inferrule
{\Gamma^\Phifun \der
    t :\tfun{\Gamma^\Phiclos}{y \of \sigma^\phi}{\tau}\\
 \Gamma^\Phiarg \der u : \sigma\\
 \Gamma, y \of \sigma \der \tau \rewto{y}{\Phiarg} \Gamma \der \tau'}
{\Gamma^{\Phifun+\Phiclos+\phi.\Phiarg} \der
    t \app u : \tau'}
\and
\inferrule
{V \der t \redto (x \mapsto v, \lam y t')\\
 V \der u \redto v_\code{arg}\\
 V\push{y \mapsto v_\code{arg}}\push{x \mapsto v}
 \der t' \redto w_1\\
 w_1 \rewto{x}{v} w_2 \rewto {y}{v_\code{arg}} w_3}
{V \der t \app u \redto w_3}
\and
\inferrule
{V' \der t \redto (x \mapsto v', \lam y t')\\
 V' \der u \redto v'_\code{arg}\\
 V\push{y \mapsto v'_\code{arg}}\push{x \mapsto v'} \der t' \redto w'_1\\
 w'_1 \rewto{x}{v'} w'_2 \rewto{y}{v'_\code{arg}} w'_3}
{V \der t \app u \redto w'_3}
\end{mathpar}

From $\Phifun \subseteq \Phifun + \Phiclos + \phi.\Phibody$ we get by
induction that
$\Gamma \der (x \mapsto v, \lam y t') =_{\Phi_0} (x \mapsto v', \lam y t') : \tfun{\Gamma^\Phiclos}{y\of\sigma^\phi}{\tau}$.
This equivalence gives us the following premises: \[
  \inferrule
     {\Gamma, x \of \rho \der v : \tau_i\\
      \Gamma^\Phi, y \of \sigma^\phi, x \of \rho^\psi \der t : \tau\\
      \Gamma^\Phi, y \of \sigma^\phi, x \of \rho^\psi \der \tau
      \rewto{x}{\Psi} \Gamma^\Phip, x \of \sigma^\phi \der \tau'\\
      \Psi \subseteq \Phi_0 \implies \Gamma \der v =_{\Phi_0} v' : \rho}
     {\Gamma \der
       (x \mapsto v, \lam y t')
       =_{\Phi_0}
       (x \mapsto v', \lam y t')
       : \tfun{\Gamma^\Phiclos}{y\of\sigma^\phi}{\tau}}
\]

If $\Psi \subseteq \Phi_0$ then $v =_{\Phi_0} v'$,
and similarly we get $v_{\code{arg}} = v'_{\code{arg}}$ only in the
case where $\Phiarg \subseteq \dots$, that is $\phi$
(the argument dependency) is $1$. In any case, $w_1$ and $w'_1$ are
$\Phi_0$-equivalent by induction, and by construction this is
preserved by the substitutions $\rewto{x}{v}$ and
$\rewto{y}{v_{\code{arg}}}$.
  \end{caselist}
\qed
\end{proof}
\end{version}

\section{Prototype Implementation}
\label{sec:impl}

To experiment with our type system, we implemented a software
prototype in OCaml. At around one thousand lines, the implementation
mainly contains two parts.
\begin{enumerate}
\item For each judgement in this paper, a definition of corresponding
  set of inference rules along with functions for building and
  checking derivations.
\item A (rudimentary) command-line interface that is based on a lexer,
  a parser, and a pretty-printer for the expressions, types, judgments
  and derivations of our system.
\end{enumerate}
For the scope checking judgments for context and types, the
implementation \emph{checks} well-scoping of the given contexts and
types. It either builds a derivation using the well-scoping rules or
fails to do so because of ill-scoped input. 

For the typing judgment, the implementation performs some
\emph{inference}.  Given a type context $\Gamma$ and an expression
$e$, it returns $\Phi$, $\sigma$, and a derivation $\Gamma^\Phi \der e
: \sigma$ if such a derivation exists.  Otherwise it fails. The
substitution and reduction judgments are deterministic and
computational in nature.  Our implementation takes the left-hand side
a judgement (with additional parameters) and \emph{computes} the
right-hand-side of the judgment along with a derivation.

Below is an example of interaction with the prototype interface:
{\small
\lstset{xleftmargin=8.5pt}
\begin{lstlisting}[]
% make
% ./closures.byte -str "let y = (y1, y2) in (y, \(x:\sigma) z)"
Parsed expression: let y = (y1, y2) in (y, $\lambda$(x:$\sigma$) z)

The variables (y1, y2, z) were unbound; we add them to the default
environment with dummy types (ty_y1, ty_y2, ty_z) and values
(val_y1, val_y2, val_z).

Inferred typing:
  y1:ty_y1$^1$,y2:ty_y2$^1$,z:ty_z$^0$ $\vdash$
    let y = (y1, y2) in (y, $\lambda$(x:$\sigma$) z)
    : ((ty_y1 * ty_y2) * [y1:ty_y1$^0$,y2:ty_y2$^0$,z:ty_z$^1$](x:$\sigma^0$) $\rightarrow$ ty_z)

Result value:
    $\!$((val_y1, val_y2), ([y1,y2,z], ((y $\mapsto \!$ (val_y1, val_y2))), $\lambda$(x) z))
\end{lstlisting}
}
In this example, adapted from the starting example of the article,
$y \of \sigma^1, z \of \tau^0 \der (y, \lam {x} z)$, one can observe
that the value $z$ is marked as non-needed by the global value
judgment, but needed in the type of the closure $\lam{x} z$. Besides,
the computed value closure has captured the local variable $y$, but
still references the variables $y1, y2$, and $z$ of the outer context.

The prototype can also produce ASCII rendering of the typing and
reduction derivations, when passed \texttt{--typing-derivation} or
\texttt{--reduction-derivation}. This can be useful in
particular in the case of typing or reduction errors, as a way to
locate the erroneous sub-derivation.

The complete source code of the prototype is available at the following URL:\\
  \url{http://gallium.inria.fr/~scherer/research/open_closure_types}

\section{Discussion}

Before we conclude, we highlight three technical points that deserve a
more in-depth discussion and that are helpful link our work to
existing and future work.

\subsubsection{Typed Closure Conversion.}

It is interesting to relate our open closure types and
typed closure conversion of Minamide et
al.~\cite{MinamideMH96}. In the classical semantics, a $\lambda$-term
$\Gamma \der \lam x e : \sigma_1 \to \sigma_2$ evaluates under the
value binding $W$ to a pair $(W, \lam x e)$, which can be
given the type $(\Gamma * (\Gamma \to \sigma_1 \to \sigma_2))$
(writing $\Gamma$ for the product of all types in the context). To
combine closures of the same observable type that capture
different environments, one needs to abstract away the
environment type by using the existential type
$\exists \rho. (\rho * (\rho \to \sigma \to \tau))$.

In our specific semantics, a closure that was originally defined in
the environment $\Gamma, \Delta$ but is then seen in the environment
$\Gamma$, only captures the values of variables in $\Delta$. Typed
closure conversion is still possible, but we would need to give it the
less abstract type
$\forall \Gamma. \exists \rho (\rho * (\Gamma \to \rho \to \sigma_1 \to \sigma_2))$.
This reflects how our open closure types allow closure types to
contain static information about variables of the current lexical
context, while still allowing free composition of closures that were
initially defined in distinct environments. Our closure types evolve
from a very open type, at the closure construction point, into the
usual ``closure conversion'' type that is completely abstract in
captured values, in the empty environment.

\subsubsection{Subtyping and Conservativity.}

As mentioned, our type system is \emph{not} conservative over the
simply-typed lambda-calculus because of the restriction on
substitution of function types (domain types must be preserved
by substitution). This is not a surprise as our types provide more
fine-grained information without giving a way to forget some of this
more precise information. Regaining conservativity is very simple.
One needs a notion
of subtyping allowing to hide variables present in closure types
(eg., $\tfun{\Gamma,\Delta}{x:\sigma_1}{\sigma_2} \leq \tfun{\Gamma}{x:\sigma_1}{\sigma_2}$
whenever $\sigma_1, \sigma_2$ are well-scoped under
$\Gamma$ alone). Systematically coercing all functions into closures
capturing the empty environment then gives us exactly the simply-typed
lambda-calculus.

\subsubsection{Polymorphism.}

We feel the two previous points could easily be formally integrated in
our work. A more important difference between our prototypical system and a
realistic framework for program analysis is the lack of polymorphism.
This could require significantly more work and is left for future
work. We conjecture that adding abstraction on type variables (and
their annotation $\phi$) is direct, but a more interesting question is
the abstraction over annotated contexts $\Gamma^\Phi$. For example, we
could want to write the following, where $\kappa$ is a formal context
variable:
\[ \der \lam f \lam x \lam y {f y x} :
  \forall \kappa \alpha \beta \gamma.
    (\tfun{\kappa}{x \of \alpha}{\tfun{\kappa}{y \of \beta}{\gamma}})
    \! \to
    (\tfun{\kappa}{y \of \beta}{\tfun{\kappa}{x \of \alpha}{\gamma}}) \]
Polymorphism seems to allow greater flexibility in the analysis of
functions taking functions as parameters.  This use of
polymorphism is related to the ``resource polymorphism'' of
\cite{Jost10}, which serves the same purpose of leaving freedom to
input functions.  Open closure
types on the other hand, are motivated by expressions that \emph{return} function
closures; the flip side of the higher-order coin. 

\section{Conclusion}
\label{sec:concl}

We have introduced open closure types and their type theory.  The
technical novelty of the type system is the ability to track
intensional properties of function application in function closures
types. To maintain this information,we have to update function types
when they escape to a smaller context. This update is performed by
a novel non-trivial substitution operation. We have proved the
soundness of this substitution and the type theory for a simply-typed
lambda calculus with pairs\Fixpoint{, let bindings and fixpoints}{ and
  let bindings}.

To demonstrate how our open closure types can be used in program
verification we have applied this technique to track data-flow
information and to ensure non-interference in the sense of
information-flow theory. We envision open closure types to be applied
in the context of type systems for strong intensional properties of
higher-order programs, and this simple system to serve as a guideline
for more advanced applications.

We already have preliminary results from an application of open
closure types in amortized resource
analysis~\cite{Jost10,HoffmannAH12}. Using them, we were for the first
time able to express a linear resource bound for the curried append
function (see Section~\ref{sec:intro}).

\subsubsection*{Acknowledgments.} This research is based on work
supported in part by DARPA CRASH grant FA8750-10-2-0254 and NSF grant
CCF-1319671.  Any opinions, findings, and conclusions contained in
this document are those of the authors and do not reflect the views of
these agencies.

\bibliographystyle{splncs03}
{
\bibliography{lit}
}

\end{document}